%% file: main.tex
\makeatletter
\@namedef{ver@mathptmx.sty}{}
\makeatother
\documentclass[
    submission,
    UKenglish,
    copyright,
    creativecommons
]{eptcs}
\providecommand{\event}{QPL 2023}

\input{preamble/packages}
\input{preamble/symbols}

\graphicspath{figures}
\title{Light-Matter Interaction in the ZXW Calculus}
\def\titlerunning{Light-Matter Interaction in the ZXW Calculus}

\author{\and
Giovanni de Felice$\null^{1}$ \and
Razin A. Shaikh$\null^{1,2}$ \and
Boldizsár Poór$\null^{1}$ \and \and
Lia Yeh$\null^{1,2}$ \and
Quanlong Wang$\null^{1}$ \and
Bob Coecke$\null^{1}$ \and
\institute{$\null^{1}$Quantinuum, 17 Beaumont Street, Oxford, OX1 2NA, United Kingdom}
\institute{$\null^{2}$University of Oxford, Oxford, United Kingdom}
}
\def\authorrunning{
  G. de Felice, R. A. Shaikh, B. Poór,\hspace*{4pt}\\
  L. Yeh, Q. Wang, and B. Coecke
}

\begin{document}

\maketitle

\begin{abstract}
    In this paper we develop a graphical calculus to rewrite photonic circuits
    involving light-matter interactions and non-linear optical effects.
    We introduce the infinite ZW calculus, a graphical language for linear operators
    on the bosonic Fock space which captures both linear and non-linear photonic
    circuits.
    This calculus is obtained by combining the QPath calculus,
    a diagrammatic language for linear optics,
    and the recently developed qudit ZXW calculus,
    a complete axiomatisation of linear maps between qudits.
    It comes with a ``lifting'' theorem allowing to prove equalities between infinite
    operators by rewriting in the ZXW calculus.
    We give a method for representing bosonic and fermionic Hamiltonians in the infinite ZW calculus.
    This allows us to derive their exponentials by diagrammatic reasoning.
    Examples include phase shifts and beam splitters, as well as non-linear Kerr media
    and Jaynes-Cummings light-matter interaction.
\end{abstract}

\section{Introduction}\label{sec:introduction}
\input{introduction.tex}

\section{Preliminaries}\label{sec:preliminaries}
\input{photonics.tex}

\subsection{ZXW calculus}\label{sec:zxw}
\input{zxw-calculus.tex}

\section{Infinite ZW calculus}\label{sec:infinite-zw}
\input{infinite-zw.tex}

\section{Photonics in ZXW}\label{sec:photonics}
\input{zxw-photonics.tex}

\section{Rewriting Hamiltonians}\label{sec:hamiltonians}
\input{hamiltonians.tex}

\section*{Acknowledgements}
We are grateful to Amar Hadzihasanovic as well as the anonymous QPL reviewers for their detailed feedback and their numerous suggestions for improvement.
We would also like to thank Richie Yeung and Alexis Toumi for insightful discussions related to this paper.
RS is supported by the Clarendon Fund Scholarship.
LY is supported by the Basil Reeve Graduate Scholarship at Oriel College with the Clarendon Fund.

\bibliographystyle{eptcs}
\bibliography{preamble/references}

\appendix

\section*{Appendix}

\section{Detailing the ZXW calculus}

\subsection{Additional notations}\label{subsec:additional-notations}
\begin{itemize}
  \item A multiplier~\cite{bonchiInteractingHopfAlgebras2017, caretteSZXCalculusScalableGraphical2019, boothCompleteZXCalculiStabiliser2022} labelled by $m$ indicates the number of connections between green and pink nodes:
  \begin{gather}
      \input{definitions/multiplier.tikz}
      \qquad \qquad
      \input{definitions/multiplier-t.tikz}
      \qquad \qquad
      \input{lemmas/multipliers/multiplier-mod.tikz}
    \tag{Mu}\label{rule:Mu}\refstepcounter{equation}
  \end{gather}

  \item The multipliers interacting together with the Z, X and W nodes play an
  important role in the ZXW calculus.
  For brevity, we use the notations $V$ with interpretation:
  \begin{gather}
    \input{definitions/v-def.tikz}
    \quad \overset{\interp{\cdot}}{\longmapsto} \quad
    \ket{0}\bra{0} + \sum_{i=1}^{d \minu 1} \ket{i}\bra{-1}
    \tag{VB}\label{rule:VB}\refstepcounter{equation}
  \end{gather}
  Furthermore, we introduce the $M$ box that equals to the following diagrams:
  \[
    \input{definitions/m-box.tikz}
  \]

  \item The Hadamard box, denoted $\input{definitions/hadamard-node.tikz}$, may be
  constructed from Z, X and W generators, as in~\cite[Axiom HD]{poorCompletenessArbitraryFinite2023}.
  We define its inverse, the yellow $H^\dagger$ box as follows:
  \begin{gather}
    \input{definitions/h-dagger.tikz}
    \tag{H$\null^\dagger$}\label{rule:HDagger}\refstepcounter{equation}
  \end{gather}
\end{itemize}

\subsection{Rules of ZXW}\label{subsec:rules-zxw}
\input{axioms.tex}

\input{app-lemmas.tex}

\end{document}

%% file: preamble/packages.tex
\usepackage[UKenglish]{babel}
\usepackage{fontenc}
\usepackage[utf8]{inputenc}
\usepackage{colorprofiles}
\usepackage[a-3b,mathxmp]{pdfx}[2018/12/22]
\usepackage{bookmark}   
\hypersetup{hidelinks}

\usepackage{csquotes}

\usepackage{etoolbox}
\apptocmd{\sloppy}{\hbadness 10000\relax}{}{} 

\usepackage{mathtools}
\usepackage{physics}
\usepackage{amsfonts}
\usepackage{amssymb}
\usepackage{amsmath}
\usepackage{mathdots}   
\usepackage{mathrsfs}   
\usepackage{bigints}    
\usepackage{xfrac}      
\usepackage{stmaryrd}   
\SetSymbolFont{stmry}{bold}{U}{stmry}{m}{n}   
\usepackage{extarrows} 

\usepackage{amsthm}
\usepackage{thmtools}
\usepackage{thm-restate}

\usepackage{placeins}

\usepackage{bm}   

\usepackage{blkarray}   

\usepackage{graphicx}
\graphicspath{figures}

\usepackage{tikzfig}
\usepackage{pgfplots} 
\pgfplotsset{compat=newest, compat/show suggested version=false} 
\usepackage{tikz-cd}

\input{preamble/zxw.tikzstyles}
\definecolor{zx_grey}{RGB}{211,211,211}
\definecolor{zx_red}{RGB}{232,165,165}
\definecolor{zx_green}{RGB}{216,248,216}

\usepackage{xtab}
\usepackage{tabu}
\usepackage{multicol}

\usepackage[capitalise, noabbrev]{cleveref}

%% file: preamble/zxw.tikzstyles
\tikzstyle{gbox}=[box, draw=black, shape=rectangle, fill={zx_green}, tikzit fill={rgb,255: red,181; green,215; blue,181}]
\tikzstyle{hbox}=[box, draw=black, shape=rectangle, fill=yellow, minimum size=.55em]
\tikzstyle{box}=[draw=black, shape=rectangle, fill={white}, minimum size=.95em, inner sep=0.15em, scale=0.85, font={\footnotesize}]

\tikzstyle{gn}=[draw=black, shape=circle, fill={zx_green}, draw=black, inner sep=0.7mm, minimum width=0pt, minimum height=0pt, tikzit fill={rgb,255: red,181; green,215; blue,181}]
\tikzstyle{rn}=[gn, fill={zx_red}, draw=black, tikzit fill={rgb,255: red,215; green,96; blue,96}]

\tikzstyle{gn_phase}=[shape=rectangle, fill={zx_green}, draw=black, minimum size=1.2em, rounded corners=0.5em, inner sep=0.2em, outer sep=-0.2em, scale=0.8, font={\footnotesize}, tikzit shape=circle, tikzit fill={rgb,255: red,181; green,215; blue,181}]
\tikzstyle{rn_phase}=[gn_phase, fill={zx_red}, draw=black, tikzit fill={rgb,255: red,215; green,96; blue,96}]

\tikzstyle{ltriang}=[rtriang, shape=isosceles triangle, fill=yellow, draw=black, shape border rotate=180]
\tikzstyle{rtriang}=[shape=isosceles triangle, fill=yellow, draw=black, isosceles triangle stretches=true, inner sep=0.8pt, minimum width=0.25cm, minimum height=2mm]
\tikzstyle{dtriang}=[rtriang, shape=isosceles triangle, fill=yellow, draw=black, shape border rotate=-90]
\tikzstyle{utriang}=[rtriang, shape=isosceles triangle, fill=yellow, draw=black, shape border rotate=90]

\tikzstyle{lmat}=[shape=signal, signal to=west, signal from=east, fill={zx_grey}, draw=black, minimum height=6pt, inner sep=1pt, font={\scriptsize  \boldmath}, tikzit fill=gray, tikzit category=GLA]
\tikzstyle{rmat}=[shape=signal, signal to=east, signal from=west, fill={zx_grey}, draw=black, minimum height=6pt, inner sep=1pt, font={\scriptsize  \boldmath}, tikzit fill=gray, tikzit category=GLA]
\tikzstyle{dmat}=[shape=signal, signal to=west, signal from=east, fill={zx_grey}, draw=black, minimum height=6pt, inner sep=1pt, font={\scriptsize  \boldmath}, tikzit fill=gray, tikzit category=GLA, rotate=90]
\tikzstyle{umat}=[shape=signal, signal to=east, signal from=west, fill={zx_grey}, draw=black, minimum height=6pt, inner sep=1pt, font={\scriptsize  \boldmath}, tikzit fill=gray, tikzit category=GLA, rotate=90]

\tikzstyle{lw}=[shape=isosceles triangle, isosceles triangle stretches=true, fill=black, draw=black, minimum width=0.4cm, minimum height=3mm, inner sep=1pt, shape border rotate=180]
\tikzstyle{rw}=[shape=isosceles triangle, isosceles triangle stretches=true, fill=black, draw=black, minimum width=0.4cm, minimum height=3mm, inner sep=1pt]
\tikzstyle{dw}=[shape=isosceles triangle, isosceles triangle stretches=true, fill=black, draw=black, minimum width=0.4cm, minimum height=3mm, inner sep=1pt, shape border rotate=-90]
\tikzstyle{uw}=[shape=isosceles triangle, isosceles triangle stretches=true, fill=black, draw=black, minimum width=0.4cm, minimum height=3mm, inner sep=1pt, shape border rotate=90]

\tikzstyle{lsplit}=[shape=isosceles triangle, isosceles triangle stretches=true, fill=white, draw=black, minimum width=0.4cm, minimum height=3mm, inner sep=1pt, shape border rotate=180]
\tikzstyle{rmerge}=[shape=isosceles triangle, isosceles triangle stretches=true, fill=white, draw=black, minimum width=0.4cm, minimum height=3mm, inner sep=1pt]
\tikzstyle{vsplit}=[shape=isosceles triangle, isosceles triangle stretches=true, fill=white, draw=black, minimum width=0.4cm, minimum height=3mm, inner sep=1pt, shape border rotate=90]

\tikzstyle{phase}=[draw=black, shape=rectangle, fill={white}, minimum size=.95em, inner sep=0.15em, scale=0.85, font={\footnotesize}]
\tikzstyle{black node}=[draw=black, shape=circle, scale= 0.3, fill={black}, font={\footnotesize}]
\tikzstyle{d_split}=[shape=trapezium, fill=white, draw=black, minimum width=11pt, inner sep=1pt]
\tikzstyle{d_merge}=[shape=trapezium, fill=white, draw=black, minimum width=11pt, inner sep=1pt, rotate=180]

\tikzstyle{bluedot}=[draw, shape=circle, fill=zx_green, scale=0.35, tikzit fill=blue]

\tikzstyle{wire label}=[font=\scriptsize, auto]

\tikzstyle{braceedge}=[-, decorate, decoration={brace, amplitude=2mm, raise=-1mm}]
\tikzstyle{blue-edge}=[-, very thick, draw=blue]
\tikzstyle{hadamard edge}=[-, dashed, draw=blue]

%% file: preamble/symbols.tex
\newcommand{\N}{\mathbb{N}}

\ifdef{\C}
  {\renewcommand{\C}{\mathbb{C}}}
  {\newcommand{\C}{\mathbb{C}}}
\newcommand{\minu}{\texttt{-}}

\newcommand{\interp}[1]{\left\llbracket#1\right\rrbracket}

\newcommand{\bR}{\begin{color}{red}}
\newcommand{\bB}{\begin{color}{blue}}
\newcommand{\bM}{\begin{color}{magenta}}
\newcommand{\bC}{\begin{color}{cyan}}
\newcommand{\bW}{\begin{color}{white}}
\newcommand{\bBl}{\begin{color}{black}}
\newcommand{\bG}{\begin{color}{green}}
\newcommand{\bY}{\begin{color}{yellow}}
\newcommand{\e}{\end{color}}

\newcommand{\tikzrefsize}[1]{\scriptsize{#1}}
\newcommand{\axref}[1]{\tikzrefsize{\eqref{rule:#1}}}
\newcommand{\lemref}[1]{\tikzrefsize{\eqref{#1}}}

\newcommand{\Mod}[1]{\ (\mathrm{mod}\ #1)}
\newtheorem{Th}{Theorem}[section]
\newtheorem{theorem}[Th]{Theorem}
\newtheorem{proposition}[Th]{Proposition} 
\newtheorem{lemma}[Th]{Lemma}

\newtheorem{definition}[Th]{Definition} 

\newtheorem{remark}[Th]{Remark}

\newcommand{\zxw}[1]{\cite[#1]{poorCompletenessArbitraryFinite2023}}
\newcommand{\inftyfy}[1]{#1$\null_{\infty}$}

%% file: introduction.tex
Graphical languages are a powerful tool for understanding, verifying and
developing software for quantum computing.
The \emph{ZX calculus}~\cite{coeckeBasicZXcalculusStudents2023, coeckePicturingQuantumProcesses2017, vandeweteringZXcalculusWorkingQuantum2020} is a graphical
language for qubit quantum computing which is currently used to solve a range of different
quantum computing tasks, including compilation~\cite{sivarajahTKetRetargetableCompiler2020}, optimization~\cite{duncanGraphtheoreticSimplificationQuantum2020, kissingerReducingNumberNonClifford2020},
machine learning~\cite{toumiDiagrammaticDifferentiationQuantum2021, wangDifferentiatingIntegratingZX2022},
measurement-based quantum computing~\cite{duncanRewritingMeasurementBasedQuantum2010, backensThereBackAgain2021},
error-correction~\cite{kissingerPhasefreeZXDiagrams2022},
and also education~\cite{coeckeQuantumPictures2022}.
It was introduced in 2007 by Coecke and Duncan~\cite{coeckeInteractingQuantumObservables2007, coeckeInteractingQuantumObservables2008, coeckeInteractingQuantumObservables2011}
to model the interaction of complementary observables.
In 2010, Coecke and Kissinger proposed to introduce the
W state~\cite{durThreeQubitsCan2000} in the calculus~\cite{coeckeCompositionalStructureMultipartite2010} and showed that
rational arithmetic operations had simple graphical representations~\cite{coeckeGHZWcalculusContains2011}.
Building on their work, Hadzihasanovic gave the first complete axiomatisation for
qubit quantum computing~\cite{hadzihasanovicDiagrammaticAxiomatisationQubit2015},
known as the \emph{ZW calculus}.
The completeness of an equational theory is important as it shows that the language is rich enough
to prove all equalities of the underlying linear maps; therefore, there are \enquote{no missing rules}.
The same techniques were subsequently used to  prove completeness of the
ZX calculus~\cite{ngUniversalCompletionZXcalculus2017, hadzihasanovicTwoCompleteAxiomatisations2018, kissingerClassicalSimulationQuantum2022, jeandelDiagrammaticReasoningClifford2018}.
The rewriting power and the ability of representing sums using W nodes motivated
the development of the \emph{ZXW calculus},
allowing the first diagrammatic treatment of Hamiltonians and exponentiation~\cite{shaikhHowSumExponentiate2022}, as well as integration and
differentiation~\cite{wangDifferentiatingIntegratingZX2022}.
A natural generalisation of ZXW to higher dimensions resulted in
the first complete axiomatisation for qudit quantum computing~\cite{poorCompletenessArbitraryFinite2023}.

Graphical languages have only recently been applied to photonic quantum computing.
The ZX calculus is now used to reason about compilation of linear optical circuits, and particularly its MBQC aspects~\cite{zilkCompilerUniversalPhotonic2022, defeliceQuantumLinearOptics2022}.
It is also used to represent error correction codes for fault-tolerant quantum computing~\cite{bombinLogicalBlocksFaultTolerant2023}.
These approaches all rely on an encoding of qubits or qubit lattices on a state
of multiple photons.
Graphical languages for reasoning about the underlying physics of interacting photons
have so far been restricted to the study of \emph{linear} optical processes.
The language of interferometers built from phases and beam splitters is well understood for single
photons~\cite{reckExperimentalRealizationAny1994, clementsOptimalDesignUniversal2016, clementLOvCalculusGraphicalLanguage2022}.
Different methods can be used to compute the amplitudes of linear optical circuits
with multiple photons~\cite{aaronsonComputationalComplexityLinear2011}.
The \emph{QPath calculus}~\cite{defeliceQuantumLinearOptics2022} is a recent graphical
approach to this problem which decomposes circuits into simpler primitive operations.
However, it has a limited rewrite system where diagrams involving multiple photons
are split into sums of terms.
Moreover, \emph{non-linear} optical phenomena appear throughout photonic quantum science.
Non-linearities are used to construct photon sources by parametric down
conversion~\cite{guoParametricDownconversionPhotonpair2017} or by emission from
a 2-level atom in a cavity~\cite{trivediGenerationNonClassicalLight2020, weinPhotonnumberEntanglementGenerated2022}.
They include mixed boson-fermion systems such as the ones studied in quantum chemistry~\cite{fitzpatrickEvaluatingLowdepthQuantum2021}.
They also include non-linear Kerr media
which allow the optical specification of universal quantum gates~\cite{azumaQuantumComputationKerrnonlinear2007},
with recently proposed graphene-based implementations~\cite{calajoNonlinearQuantumLogic2023}.
These non-linear effects rely on different forms of light-matter interaction.
Hence, we need a graphical calculus that could represent interaction between bosons and fermions.

The semantics of a photonic graphical calculus is based on infinite dimensional Fock space.
However, graphical calculi with infinite dimensional semantics do not have a strong reasoning system.
Even simple equalities, such as the snake equation, are not well-defined.
To alleviate this problem, we truncate the infinite dimensional vector space to finite dimensions for which we have a complete graphical calculus, i.e.\@ the ZXW calculus.
Then we prove a lifting theorem to show that, under mild conditions, the results derived in this truncated graphical calculus also hold in infinite dimensional vector space.
This idea is similar to that of Gogioso and Genovese~\cite{gogiosoInfinitedimensionalCategoricalQuantum2016, gogiosoQuantumFieldTheory2018, gogiosoQuantumFieldTheory2019}
where the transfer theorem from non-standard analysis is used to prove results about
infinite dimensional Hilbert spaces.
However, by restricting our attention to states with finite support and maps preserving this property, our lifting theorem admits a simpler proof which does not require non-standard analysis.

Our approach of leveraging a finite dimensional calculus has a key advantage over constraining to solely infinite dimensional formalisms.
We can now reason about not only bosons, but also other finite dimensional systems such as fermions, altogether in the same framework.
Hence, we achieve a unified calculus powerful enough to reason about light-matter interactions.

This paper provides a framework for representing and rewriting quantum optics in the ZXW calculus.
We start by introducing the QPath calculus and the ZXW calculus (\cref{sec:preliminaries}).
The first is a graphical language for linear optics~\cite{defeliceQuantumLinearOptics2022},
describing the behaviour of bosonic maps in the category $\bf{Vect}_\mathbb{N}$
of linear operators on the bosonic Fock space.
The second is a complete axiomatisation $\interp{\cdot}_d: \bf{ZXW}_d \xrightarrow{\sim} \bf{Vect}_d$
of the category of linear maps between qudits~\cite{poorCompletenessArbitraryFinite2023}.
In \cref{sec:infinite-zw} and \cref{sec:photonics} we present our main contributions:
(1) we introduce a more expressive calculus $\bf{ZW}_\infty$ for the category of
linear operators on the Fock space $\bf{Vect}_\mathbb{N}$,
(2) we give a truncated interpretation $\mathcal{T}_d : \bf{ZW}_\infty \to \bf{ZXW}_d$
and prove \cref{th:lifting} which allows to derive equalities in $\bf{ZW}_\infty$
by lifting them from $\bf{ZXW}_d$ for \enquote{big enough} $d$,
(3) we give an axiomatisation of $\bf{ZW}_\infty$ for which we prove soundness via
the lifting theorem.
We end by exploring different applications of this graphical calculus in the
field of non-linear optics (\cref{sec:hamiltonians}).
We give a general method for representing quantum optical Hamiltonians in $\bf{ZW}_\infty$,
including Kerr media and Jaynes-Cummings light-matter interaction.
Building on~\cite{shaikhHowSumExponentiate2022}, we show how to exponentiate
these Hamiltonians using diagrammatic techniques.

%% file: photonics.tex
\setlength{\abovedisplayskip}{0pt}
\setlength{\belowdisplayskip}{0pt}

\subsection{Fock space}

Consider the state space of a single bosonic mode $\interp{1} = \bigoplus_{n=0}^\infty \mathbb{C}$.
An element in this space is usually denoted as an infinite sum
$\ket{a} = \sum_{n = 0}^\infty a_n \ket{n}$.
Depending on the chosen boundedness condition for these sequences,
we get different classes of operators acting on $\interp{1}$.
In their work on the quantum harmonic oscillator, Vicary~\cite{vicaryCategoricalFrameworkQuantum2008} imposed
a strong normalising condition on sequences, of the form $\sum_nc^na_n < \infty$
for any $c \in \mathbb{C}$, interpreting $\bigoplus$ as an infinite biproduct.
The valid maps are those that preserve this normalising condition.
They include phases, coherent states and creation operators.
In this paper, we interpret $\bigoplus$ as an infinite direct sum.
The valid states are sequences with finitely many non-zero terms.
This means we do not allow coherent states in our spaces.
The operators that we consider are those that send finite states to finite states;
they form a category $\bf{Vect}_\mathbb{N}$.

For a system of $n$ bosonic particles occupying $m$ possible modes,
the state space is given by the Fock space defined as follows.
\[
  \interp{m} = \bigoplus_{n=0}^\infty (\mathbb{C}^m)^{\tilde{\otimes} n}
  \simeq \interp{1}^{\otimes m}
\]
where $\tilde{\otimes}$ is the symmetrised tensor product, i.e.\@ the quotient of the
tensor product given by identifying $x = x_1 \otimes \cdots \otimes x_n$ with
$y = y_1 \otimes \cdots \otimes y_n$ whenever there is a permutation $\sigma$ such
that $\sigma(x) = y$. The isomorphism $\interp{m} \simeq \interp{1}^{\otimes m}$ is given
by counting the number of photons occupying each mode~\cite[Proposition 3.2]{defeliceQuantumLinearOptics2022}.
Let $\bf{Vect}$ be the symmetric monoidal category of vector spaces
and linear maps with tensor product $\otimes$. Symmetrised spaces of the form
$(\mathbb{C}^m)^{\tilde{\otimes} n}$ are objects of this category.
We define the infinite direct sum $\bigoplus_{n=0}^\infty A_n$ for
$A_n \in \bf{Vect}$ as the set of sequences $(a_0, a_1, \dots)$ with $a_n \in A_n$
such that $a_n = 0$ except for finitely many $n$. This means we only allow states
$\ket{a} = \sum_{n = 0}^{k} a_n \ket{n} \in \interp{1}$ with a finite number $k$
of particles.
We can now build our base category $\bf{Vect}_\mathbb{N}$ of linear maps on the Fock space:
objects are finite tensor products of the complex space $\interp{1}$,
morphisms are linear maps $f: \interp{m} \simeq (\interp{1})^{\otimes m} \to  (\interp{1})^{\otimes k} \simeq \interp{k}$.
Note that we allow unbounded linear maps as long as their domain and codomain
are in the infinite direct sum. As argued by Vicary~\cite[Section 6]{vicaryCategoricalFrameworkQuantum2008},
these are necessary to understand the algebraic structure of the Fock space.

Note that $\bf{Vect}_\mathbb{N}$ is not compact closed.
In fact, generalising the standard cups and caps would give a state without a
fixed number of particles $\sum_n \ket{nn} \notin \bf{Vect}_\mathbb{N}$.
$\bf{Vect}_\mathbb{N}$ admits a dagger structure only on a subclass of morphisms.
Indeed, the valid map $\epsilon: \interp{1} \to 1$, with $\epsilon \ket{n} = 1$
for all $n$, gives an invalid state $\sum_n^\infty \ket{n}$ when taking the adjoint.
However, the dagger is well-defined for most morphisms of interest, and these
are closed under composition.
It is also easy to see that $\bf{Vect}_\mathbb{N}$ is enriched in weighted finite sums
over the complex numbers.
We now consider an interesting subclass of processes in $\bf{Vect}_\mathbb{N}$,
the ones that may be performed using linear optics and creation/annihilation of single particles.

\subsection{Bosonic nodes}\label{sec:qpath}
\setlength{\abovedisplayskip}{5pt}
\setlength{\belowdisplayskip}{5pt}

Bosonic nodes, or split and merge maps, are linear operators on the Fock space
exhibiting a bialgebra structure, equivalent to the binomial bialgebra on
polynomials \cite{mullinFoundationsCombinatorialTheory1970, figueroaCombinatorialHopfAlgebras2005}.
In categorical quantum mechanics, they were first studied by Vicary and Fiore~\cite{vicaryCategoricalFrameworkQuantum2008, fioreAxiomaticsCombinatorialModel2015}
who characterised the map $b$ that copies coherent states of the form
$\ket{\alpha} = \sum_{n}\frac{\alpha^n}{\sqrt{n!}} \ket{n}$, in the sense that
$b \ket{\alpha} = \ket{\alpha} \otimes \ket{\alpha}$.
Coherent states are usually defined as eigenvectors of the bosonic creation operator
$a^\dagger\ket{\alpha} = \alpha \ket{\alpha}$. They account for Poissonian
distributions of photons from coherent light sources such as lasers~\cite{fujiiIntroductionCoherentStates2002}.
The split map was later studied as an anyonic generalisation of the W algebra
in~\cite{hadzihasanovicAlgebraEntanglementGeometry2017} and as a generator for Feynman diagrams
in~\cite{shaikhCategoricalSemanticsFeynman2022}.

The QPath calculus~\cite{defeliceQuantumLinearOptics2022} is graphical calculus
for linear optics based on the bosonic split and merge maps.
A wire in the QPath calculus corresponds to a bosonic mode and is thus interpreted
as the space $\interp{1}$ with basis given by occupation numbers $\ket{n}$ for $n \in \mathbb{N}$.
The interpretation of $m$ wires is obtained using the tensor product
$\interp{m} \simeq \interp{1}^{\otimes m}$.
The main generator of QPath is the bosonic \emph{split} map, defined as follows:
\[
  \input{qpath/generators/split.tikz}
  \quad \xmapsto{\interp{\cdot}} \quad
  \ket{n} \mapsto \sum_{k=0}^n \binom{n}{k}^{\frac{1}{2}} \ket{k} \ket{n-k}
\]
Its adjoint or dagger called \emph{merge}, is also a generator:
\[
  \input{qpath/generators/merge.tikz}
  \quad \xmapsto{\interp{\cdot}} \quad
  \ket{n, m} \mapsto \binom{n + m}{n}^{\frac{1}{2}} \ket{n+m}
\]
We also have \emph{endomorphisms} for any $r \in \mathbb{C}$:
\[
  \input{qpath/generators/phase.tikz}
  \quad \xmapsto{\interp{\cdot}} \quad
  Z(\vec{r}): \ket{n} \mapsto r^n \ket{n}
\]
or phase shifts when $r= e^{i\theta}$, then the symmetry, or \emph{swap}:
\[
  \input{qudit-generators/swap.tikz}
  \quad \xmapsto{\interp{\cdot}} \quad
  \ket{n, m} \mapsto \ket{m, n}
\]
and, $n$-photon states and effects:
\[
  \input{qpath/generators/states.tikz} \quad \xmapsto{\interp{\cdot}} \quad \ket{n}
  \qquad \qquad \qquad
  \input{qpath/generators/effects.tikz} \quad \xmapsto{\interp{\cdot}} \quad \bra{n}
\]
All the maps defined above send basis states to finite sums of basis states and
are therefore valid morphisms in $\bf{Vect}_\mathbb{N}$.
These generators satisfy the axioms of a bialgebra with
endomorphisms, along with additional rules involving occupied photon states~\cite{defeliceQuantumLinearOptics2022}.

%% file: figures/qpath/generators/split.tikz
\begin{tikzpicture}
	\begin{pgfonlayer}{nodelayer}
		\node [style=none] (0) at (-1.25, 0) {};
		\node [style=lsplit] (1) at (0, 0) {};
		\node [style=none] (2) at (1.5, 0.75) {};
		\node [style=none] (3) at (1.5, -0.75) {};
	\end{pgfonlayer}
	\begin{pgfonlayer}{edgelayer}
		\draw [in=45, out=180] (2.center) to (1);
		\draw [in=-180, out=-45] (1) to (3.center);
		\draw (1) to (0.center);
	\end{pgfonlayer}
\end{tikzpicture}

%% file: figures/qpath/generators/merge.tikz
\begin{tikzpicture}
	\begin{pgfonlayer}{nodelayer}
		\node [style=none] (0) at (1.25, 0) {};
		\node [style=rmerge] (1) at (0, 0) {};
		\node [style=none] (2) at (-1.5, 0.75) {};
		\node [style=none] (3) at (-1.5, -0.75) {};
	\end{pgfonlayer}
	\begin{pgfonlayer}{edgelayer}
		\draw [in=135, out=0] (2.center) to (1);
		\draw [in=0, out=-135] (1) to (3.center);
		\draw (1) to (0.center);
	\end{pgfonlayer}
\end{tikzpicture}

%% file: figures/qpath/generators/phase.tikz
\begin{tikzpicture}
	\begin{pgfonlayer}{nodelayer}
		\node [style=none] (0) at (-1.25, 0) {};
		\node [style=none] (1) at (1.25, 0) {};
		\node [style=phase] (2) at (0, 0) {$r$};
	\end{pgfonlayer}
	\begin{pgfonlayer}{edgelayer}
		\draw (0.center) to (2);
		\draw (2) to (1.center);
	\end{pgfonlayer}
\end{tikzpicture}

%% file: figures/qudit-generators/swap.tikz
\begin{tikzpicture}
	\begin{pgfonlayer}{nodelayer}
		\node [style=none] (0) at (-0.75, 0.5) {};
		\node [style=none] (1) at (0.75, 0.5) {};
		\node [style=none] (2) at (0.75, -0.5) {};
		\node [style=none] (3) at (-0.75, -0.5) {};
		\node [style=none] (4) at (-1, 0.5) {};
		\node [style=none] (5) at (-1, -0.5) {};
		\node [style=none] (6) at (1, -0.5) {};
		\node [style=none] (7) at (1, 0.5) {};
	\end{pgfonlayer}
	\begin{pgfonlayer}{edgelayer}
		\draw (5.center) to (3.center);
		\draw (4.center) to (0.center);
		\draw (7.center) to (1.center);
		\draw (6.center) to (2.center);
		\draw [in=0, out=-180, looseness=0.75] (2.center) to (0.center);
		\draw [in=0, out=180, looseness=0.75] (1.center) to (3.center);
	\end{pgfonlayer}
\end{tikzpicture}

%% file: figures/qpath/generators/states.tikz
\begin{tikzpicture}
	\begin{pgfonlayer}{nodelayer}
		\node [style=black node] (0) at (-1, 0) {};
		\node [style=none] (2) at (0.5, 0) {};
		\node [style=none] (5) at (-1, 0.5) {$\scriptstyle n$};
	\end{pgfonlayer}
	\begin{pgfonlayer}{edgelayer}
		\draw (0) to (2.center);
	\end{pgfonlayer}
\end{tikzpicture}

%% file: figures/qpath/generators/effects.tikz
\begin{tikzpicture}
	\begin{pgfonlayer}{nodelayer}
		\node [style=black node] (0) at (0.5, 0) {};
		\node [style=none] (1) at (-1, 0) {};
		\node [style=none] (2) at (0.5, 0.5) {$\scriptstyle n$};
	\end{pgfonlayer}
	\begin{pgfonlayer}{edgelayer}
		\draw (1.center) to (0);
	\end{pgfonlayer}
\end{tikzpicture}

%% file: zxw-calculus.tex
The ZXW calculus is a complete diagrammatic language for qudit quantum
computing~\cite{poorCompletenessArbitraryFinite2023}.
Formally, let $\bf{Vect}_d$ be the symmetric monoidal category with objects tensor
products of qudits $\mathbb{C}^d$ and morphisms given by linear maps between them.
Completeness means that the interpretation
$\interp{\cdot}_d: \bf{ZXW}_d \to \bf{Vect}_d$ is an isomorphism.
We describe the generators of the ZXW calculus along with their interpretation.
Throughout the section, we fix a natural number $d$ and use the arbitrary complex
vectors $\overrightarrow{a} = (a_1,\dotsc, a_{d-1})$ and $\overrightarrow{b} = (b_1,\dotsc, b_{d-1})$.
Note that here we made a slight change for the presentation of the qudit ZXW calculus
in comparison to the version of~\cite{poorCompletenessArbitraryFinite2023}:
we use the X spider instead of the Hadamard node as a generator now, so the
definition of the former becomes a rule and the rule for the latter turns into a definition.
Also we removed the (H1) rule of~\cite{poorCompletenessArbitraryFinite2023} from
the rule set since we find it can be derived from other rules now.

\subsubsection{Generators of ZXW}\label{subsec:generators_n_interpretation}

The generators of $\bf{ZXW}_d$ together with their standard interpretation $\interp{\cdot}_d$ are:
\begin{itemize}
  \item The \emph{Z spider},
  \[
    \input{qudit-generators/generalgreenspiderqdit2.tikz}
    \quad \overset{\interp{\cdot}_d}{\longmapsto} \quad
    \sum_{j=0}^{d-1}a_j\ket{j}^{\otimes m}\bra{j}^{\otimes n}, \quad \text{where } a_0 = 1.
  \]

  \item The \emph{X spider}, with parameter $j$ which can be taken modulo $d$,
  \[
    \input{qudit-generators/quditrspiderclassicnm.tikz}
    \quad \overset{\interp{\cdot}_d}{\longmapsto} \hspace{-.3cm}
    \sum_{\substack{
      0 \leq i_1, \cdots, i_m,  j_1, \cdots, j_n \leq d-1 \\
      i_1+\cdots+ i_m+j \equiv j_1+\cdots +j_n \Mod{d}}
    } \hspace{-1.75cm} \ket{i_1, \cdots, i_m}\bra{j_1, \cdots, j_n},
  \]

  \item The \emph{W node},
  \[
    \input{definitions/w1ton.tikz}
    \quad \overset{\interp{\cdot}_d}{\longmapsto} \quad
    \ket{0\cdots0}\bra{0} + \sum_{i=1}^{d-1}(\ket{i0\cdots 00}+\cdots +\ket{00\cdots 0i})\bra{i}
  \]

  \item The \emph{swap},
  \[
    \input{qudit-generators/swap.tikz}
    \quad \overset{\interp{\cdot}_d}{\longmapsto} \quad
    \sum_{i, j=0}^{d-1}\ket{ji}\bra{ij}.
  \]
\end{itemize}

\subsubsection{Notations}

For convenience, we introduce the following notation which will be used throughout the paper:
\begin{itemize}
  \item The phase depicted by a green circle spider can be defined using the Z box:
  \[
    \input{definitions/circlegspiders1.tikz}
  \]
  where
  $\overrightarrow{\alpha} = (\alpha_1, \cdots, \alpha_{d-1})$,\,
  $e^{i\overrightarrow{\alpha}}=(e^{i\alpha_1}, \cdots, e^{i\alpha_{d-1}})$,\, and
  $\alpha_i \in [0, 2\pi)$.

  \item The cup and cap, i.e.\@ the qudit Bell state and its transpose, are defined as follows:
  \begin{gather}
    \input{definitions/compactstructures1.tikz}
    \qquad \qquad
    \input{definitions/compactstructures2.tikz}
    \tag{S3}\label{rule:S3}\refstepcounter{equation}
  \end{gather}

  \item We use a yellow $D$ box to denote the \emph{dualiser}, with the given interpretation:
  \begin{gather}
    \input{definitions/dualiser.tikz}
    \quad \overset{\interp{\cdot}_d}{\longmapsto} \quad
    \sum_{i = 0}^{d \minu 1} \ket{i} \bra{d-i}.
    \tag{Du}\label{rule:Du}\refstepcounter{equation}
  \end{gather}

  \item It is useful to define the yellow triangle in terms of the $W$ node and green spider:
  \begin{gather}
    \input{definitions/trianlge-def.tikz}
    \quad \overset{\interp{\cdot}_d}{\longmapsto} \quad
    I_d + \sum_{i=1}^{d \minu 1} \ket{0} \bra{i}
    \tag{YT}\label{rule:YT}\refstepcounter{equation}
  \end{gather}
\end{itemize}
Furthermore, some additional notations are presented in \cref{subsec:additional-notations}.

We refer to \cref{subsec:rules-zxw} for the rules of the calculus.

%% file: figures/qudit-generators/generalgreenspiderqdit2.tikz
\begin{tikzpicture}
	\begin{pgfonlayer}{nodelayer}
		\node [style=none] (0) at (1.875, 0.875) {};
		\node [style=none] (2) at (1.875, -0.875) {};
		\node [style=none] (5) at (-1.875, 0.875) {};
		\node [style=none] (9) at (-1.875, -0.875) {};
		\node [style=none] (13) at (2.375, 0) {{\scriptsize $m$}};
		\node [style=none] (14) at (-2.375, 0) {{\scriptsize $n$}};
		\node [style=none] (15) at (-1.25, 0.25) {$\vdots$};
		\node [style=none] (16) at (1.25, 0.25) {$\vdots$};
		\node [style=none] (17) at (1.5, 0.75) {};
		\node [style=none] (19) at (1.5, -0.75) {};
		\node [style=gbox] (20) at (0, 0) {$\overrightarrow{a}$};
		\node [style=none] (21) at (-1.5, 0.75) {};
		\node [style=none] (23) at (-1.5, -0.75) {};
	\end{pgfonlayer}
	\begin{pgfonlayer}{edgelayer}
		\draw [style=braceedge] (0.center) to (2.center);
		\draw [style=braceedge] (9.center) to (5.center);
		\draw [style=none, bend right=15] (17.center) to (20);
		\draw [style=none, bend right=15] (20) to (21.center);
		\draw [style=none, bend left=15] (19.center) to (20);
		\draw [style=none, bend left=15] (20) to (23.center);
	\end{pgfonlayer}
\end{tikzpicture}

%% file: figures/qudit-generators/quditrspiderclassicnm.tikz
\begin{tikzpicture}
	\begin{pgfonlayer}{nodelayer}
		\node [style=none] (47) at (1.875, 0.875) {};
		\node [style=none] (48) at (1.875, -0.875) {};
		\node [style=none] (49) at (-1.875, 0.875) {};
		\node [style=none] (50) at (-1.875, -0.875) {};
		\node [style=none] (51) at (2.375, 0) {{\scriptsize $m$}};
		\node [style=none] (52) at (-2.375, 0) {{\scriptsize $n$}};
		\node [style=none] (53) at (-1.25, 0.25) {$\vdots$};
		\node [style=none] (54) at (1.25, 0.25) {$\vdots$};
		\node [style=none] (55) at (1.5, 0.75) {};
		\node [style=none] (56) at (1.5, -0.75) {};
		\node [style={rn_phase}] (57) at (0, 0) {$K_j$ };
		\node [style=none] (58) at (-1.5, 0.75) {};
		\node [style=none] (59) at (-1.5, -0.75) {};
	\end{pgfonlayer}
	\begin{pgfonlayer}{edgelayer}
		\draw [style=braceedge] (47.center) to (48.center);
		\draw [style=braceedge] (50.center) to (49.center);
		\draw [style=none, bend right=15] (55.center) to (57);
		\draw [style=none, bend right=15] (57) to (58.center);
		\draw [style=none, bend left=15] (56.center) to (57);
		\draw [style=none, bend left=15] (57) to (59.center);
	\end{pgfonlayer}
\end{tikzpicture}

%% file: infinite-zw.tex
\begingroup

\allowdisplaybreaks

We introduce the graphical language $\bf{ZW}_\infty$ for reasoning about linear
operators on the Fock space.
By setting $d=\infty$, we generalise the qudit interpretation of each
ZXW generator to obtain novel maps on the Fock space.
Combining this with the QPath calculus, we obtain a (partial) axiomatisation
of linear operators in $\bf{Vect}_\mathbb{N}$.

\subsection{Generators}

The \emph{Z spider} may be interpreted as the following operator in $\bf{Vect}_\mathbb{N}$:
\begin{equation}
  \label{eq:inf-z-spider}
  \input{qudit-generators/generalgreenspiderqdit2.tikz}
  \quad \overset{\interp{\cdot}}{\longmapsto} \quad
  \sum_{i=0}^\infty a_i \ket{i}^{\otimes m}\bra{i}^{\otimes n},
  \quad \text{where } n > 0, \overrightarrow{a}=(a_1,\cdots, a_k, \cdots),\ a_0 \coloneqq 1.
\end{equation}\refstepcounter{equation}
Note that we require the number of inputs to be greater than $0$.
In fact, the unit Z spider $\sum_{n=0}^\infty \ket{n}$ does not have finite support
and is thus not allowed in $\bf{ZW}_\infty$.
In particular, we do not have cups because these are defined in terms of unit Z spiders.
However, we do have caps and counit Z spiders since these map finite states to finite states.
Moreover, Z boxes where $n=0$ and $\overrightarrow{b}$ has a \emph{finite support} are allowed as generators, that is, for some finite $N \in \N$:
\begin{equation}
  \label{eq:inf-z-state}
  \input{qudit-generators/greenspiderstate.tikz}
  \quad \overset{\interp{\cdot}}{\longmapsto} \quad
  \sum_{i=0}^N a_i \ket{i}^{\otimes m},
  \quad \text{where } n > 0, \overrightarrow{a}=(a_1,\cdots, a_N, 0, \cdots),\ a_0 \coloneqq 1.
\end{equation}\refstepcounter{equation}
Note that QPath endomorphisms are Z boxes of the form:
\begin{gather*}
  \input{qpath/generators/phase.tikz}
  \qquad = \qquad
  \input{qpath/truncated/phase-def.tikz}
\end{gather*}
where $r^{\overrightarrow{N}} = (r^1,\,r^2,\,\cdots,\,r^{k}, \cdots)$.
The \emph{W node} is straightforwardly generalised to infinite dimensions:
\[
  \input{definitions/w1ton.tikz}
  \quad \overset{\interp{\cdot}}{\longmapsto} \quad
  \ket{0\cdots0}\bra{0} + \sum_{i=1}^{\infty}(\ket{i0\cdots 0}+\cdots +\ket{0\cdots 0i})\bra{i}
\]
We also take the dagger of W as a generator. One may check that both maps preserve finite states.
Let us now consider the \emph{X spider} from the ZXW calculus. Its action on the
basis in $\bf{ZXW}_d$ is given by addition modulo $d$. We can consider
addition in $\mathbb{N}$ as a natural generalisation for the X spider in the Fock space.
However, this does not yield a group structure on $\mathbb{N}$ and the Frobenius law fails
in infinite dimensions. Moreover, we already have a node that performs addition
of natural numbers, the \emph{split map}:
\[
  \input{qpath/generators/split.tikz}
  \quad \xmapsto{\interp{\cdot}} \quad
  \sum_{n=0}^\infty \sum_{k=0}^n \binom{n}{k}^{\frac{1}{2}} \ket{k, n-k}\bra{n}
\]
We also take \emph{merge} as a generator of $\bf{ZW}_\infty$.
Instead of a Frobenius structure we get a bialgebra between merge and split.
The binomial coefficients in fact ensure that the bialgebra law holds.
However, this also means that the interaction of the split map with Z spiders
is a bialgebra only up to coefficients.
We also take the $n$-particle creation $\input{qpath/generators/state-1.tikz}$
and annihilation $\input{qpath/generators/effect-1.tikz}$ as generators of $\bf{ZW}_\infty$.

\subsection{Axioms}

The axioms of $\bf{ZW}_\infty$ are obtained by merging the rules of QPath and those
of ZXW governing the Z and W nodes.
We only give a partial axiomatisation with the rules that we use in \cref{sec:hamiltonians}.
The Z spider satisfies the axioms of a non-unital
special Frobenius algebra, an algebraic structure which characterises bases in
infinite dimensional vector spaces~\cite{abramskyAlgebrasNonunitalFrobenius2012}.
The rules of the W algebra as well as Axioms \eqref{rule:iBsj} and \eqref{rule:iK0} are directly lifted from qudit ZXW.
Bosonic nodes form a bialgebra with a semiring of endomorphisms.
All the rules of QPath from~\cite{defeliceQuantumLinearOptics2022} hold in $\bf{ZW}_\infty$,
here we only give the ones we use.
We moreover give a set of interaction rules relating Z, W and bosonic nodes.
Rule \eqref{rule:iW1} shows the behaviour of the W algebra.
\eqref{rule:bW1} relates the split map and W states and replaces the ``branching''
rule of QPath.
The bialgebra law between bosonic nodes and Z spiders only holds up to factorial coefficients \eqref{rule:bZBA}.
It holds on the nose if the binomial coefficients are removed
from the definition of bosonic nodes. However then, the bialgebra law between
split and merge would fail.
We also have a new version of the ZXW trialgebra law \eqref{rule:bTA} with bosonic nodes replacing X spiders, proved in \cref{prop:trialgebra}.

\subsubsection*{Non-unital Frobenius algebra}

\setlength{\abovedisplayskip}{0pt}
\setlength{\belowdisplayskip}{0pt}
\begin{gather*}
  \input{axioms/gengspiderfusedit.tikz}
  \tag{\inftyfy{S1}}\label{rule:iS1}\refstepcounter{equation}
\end{gather*}

\subsubsection*{Rules generalised from qudit ZW}
\begin{multicols}{2}
  \setlength{\abovedisplayskip}{0pt}
  \setlength{\belowdisplayskip}{0pt}
  \noindent
  \begin{gather}
    \input{axioms/w-bialgebra.tikz}
    \tag{\inftyfy{BZW}}\label{rule:iBZW}\refstepcounter{equation} \\
    \input{axioms/phasecopydit.tikz}
    \tag{\inftyfy{Pcy}}\label{rule:iPcy}\refstepcounter{equation} \\
    \input{axioms/additiondit.tikz}
    \tag{\inftyfy{AD}}\label{rule:iAD}\refstepcounter{equation} \\ \columnbreak
    \input{axioms/wsymetrydit.tikz}
    \tag{\inftyfy{Sym}}\label{rule:iSym}\refstepcounter{equation} \\
    \input{axioms/associatedit.tikz}
    \tag{\inftyfy{Aso}}\label{rule:iAso}\refstepcounter{equation} \\
    \input{axioms/w-w-algebra.tikz}
    \tag{\inftyfy{WW}}\label{rule:iWW}\refstepcounter{equation}
  \end{gather}
\end{multicols}

\subsubsection*{Rules from QPath}
\begin{multicols}{2}
  \setlength{\abovedisplayskip}{0pt}
  \setlength{\belowdisplayskip}{0pt}
  \noindent
  \begin{gather}
    \input{qpath/axioms/split-sym.tikz}
    \tag{bSym}\label{rule:bSym}\refstepcounter{equation} \\
    \input{qpath/axioms/split-aso.tikz}
    \tag{bAso}\label{rule:bAso}\refstepcounter{equation} \\
    \input{qpath/axioms/split-merge-bialgebra.tikz}
    \tag{bBA}\label{rule:bBA}\refstepcounter{equation} \\
    \input{qpath/axioms/split-id.tikz}
    \tag{bId}\label{rule:bId}\refstepcounter{equation}
  \end{gather}
\end{multicols}

\subsubsection*{Interaction rules}
\begingroup
\begin{multicols}{2}
  \setlength{\abovedisplayskip}{0pt}
  \setlength{\belowdisplayskip}{0pt}
  \noindent
  \begin{gather}
    \input{infinite-zw/bialgebra-coeffs.tikz}
    \tag{bZBA}\label{rule:bZBA}\refstepcounter{equation} \\
    \scalebox{1}{\input{infinite-zw/trialgebraAZW.tikz}}
    \tag{bTA}\label{rule:bTA}\refstepcounter{equation} \\
    \input{infinite-zw/branching-new.tikz}
    \tag{bW1}\label{rule:bW1}\refstepcounter{equation} \\
    \input{infinite-zw/axiom-k0.tikz}
    \tag{\inftyfy{K0}}\label{rule:iK0}\refstepcounter{equation} \\
    \input{lemmas/wone.tikz}
    \tag{\inftyfy{W1}}\label{rule:iW1}\refstepcounter{equation} \\
    \input{infinite-zw/axiom-bsj.tikz}
    \tag{\inftyfy{Bsj}}\label{rule:iBsj}\refstepcounter{equation}
  \end{gather}
  \vspace{\jot}
  \begin{align*}
    \text{where}&\quad
    e_n=(\underbrace{0,\dotsc, 1}_{n}, 0, \dotsc)
  \end{align*}
  \vspace*{-1cm}
\end{multicols}
\endgroup

\endgroup

%% file: figures/qudit-generators/greenspiderstate.tikz
\begin{tikzpicture}
	\begin{pgfonlayer}{nodelayer}
		\node [style=none] (0) at (1.125, 0.875) {};
		\node [style=none] (2) at (1.125, -0.875) {};
		\node [style=none] (13) at (1.625, 0) {{\scriptsize $m$}};
		\node [style=none] (16) at (0.5, 0.25) {$\vdots$};
		\node [style=none] (17) at (0.75, 0.75) {};
		\node [style=none] (19) at (0.75, -0.75) {};
		\node [style=gbox] (20) at (-0.75, 0) {$\overrightarrow{b}$};
	\end{pgfonlayer}
	\begin{pgfonlayer}{edgelayer}
		\draw [style=braceedge] (0.center) to (2.center);
		\draw [style=none, bend right=15] (17.center) to (20);
		\draw [style=none, bend left=15] (19.center) to (20);
	\end{pgfonlayer}
\end{tikzpicture}

%% file: figures/qpath/truncated/phase-def.tikz
\begin{tikzpicture}
	\begin{pgfonlayer}{nodelayer}
		\node [style=none] (1) at (-2, 0) {};
		\node [style=none] (2) at (2, 0) {};
		\node [style=gbox] (3) at (0, 0) {$r^{\overrightarrow{N}}$};
	\end{pgfonlayer}
	\begin{pgfonlayer}{edgelayer}
		\draw (1.center) to (3);
		\draw (3) to (2.center);
	\end{pgfonlayer}
\end{tikzpicture}

%% file: figures/qpath/generators/state-1.tikz
\begin{tikzpicture}
	\begin{pgfonlayer}{nodelayer}
		\node [style=none] (0) at (0.75, 0) {};
		\node [style=black node] (1) at (-0.5, 0) {};
		\node [style=none] (2) at (-0.5, 0.5) {$\scriptstyle n$};
	\end{pgfonlayer}
	\begin{pgfonlayer}{edgelayer}
		\draw (1) to (0.center);
	\end{pgfonlayer}
\end{tikzpicture}

%% file: figures/qpath/generators/effect-1.tikz
\begin{tikzpicture}
	\begin{pgfonlayer}{nodelayer}
		\node [style=none] (0) at (-0.5, 0) {};
		\node [style=black node] (1) at (0.75, 0) {};
		\node [style=none] (2) at (0.75, 0.5) {$\scriptstyle n$};
	\end{pgfonlayer}
	\begin{pgfonlayer}{edgelayer}
		\draw (1) to (0.center);
	\end{pgfonlayer}
\end{tikzpicture}

%% file: figures/qpath/axioms/split-sym.tikz
\begin{tikzpicture}
	\begin{pgfonlayer}{nodelayer}
		\node [style=lsplit] (0) at (-2.25, 0) {};
		\node [style=none] (1) at (-1.25, -0.75) {};
		\node [style=none] (2) at (-1.25, 0.75) {};
		\node [style=none] (3) at (-1, -0.75) {};
		\node [style=none] (4) at (-1, 0.75) {};
		\node [style=none] (5) at (-3.5, 0) {};
		\node [style=lsplit] (6) at (2.25, 0) {};
		\node [style=none] (7) at (3.25, -0.75) {};
		\node [style=none] (8) at (3.25, 0.75) {};
		\node [style=none] (9) at (3.25, -0.75) {};
		\node [style=none] (10) at (3.25, 0.75) {};
		\node [style=none] (11) at (1, 0) {};
		\node [style=none] (12) at (4.25, -0.75) {};
		\node [style=none] (13) at (4.25, 0.75) {};
		\node [style=none] (14) at (4.5, 0.75) {};
		\node [style=none] (15) at (4.5, -0.75) {};
		\node [style=none] (16) at (0, 0) {$=$};
	\end{pgfonlayer}
	\begin{pgfonlayer}{edgelayer}
		\draw [in=180, out=-45, looseness=0.75] (0) to (1.center);
		\draw [in=45, out=180, looseness=0.75] (2.center) to (0);
		\draw (4.center) to (2.center);
		\draw (3.center) to (1.center);
		\draw (5.center) to (0);
		\draw [in=180, out=-45, looseness=0.75] (6) to (7.center);
		\draw [in=45, out=180, looseness=0.75] (8.center) to (6);
		\draw (10.center) to (8.center);
		\draw (9.center) to (7.center);
		\draw (11.center) to (6);
		\draw (12.center) to (15.center);
		\draw (13.center) to (14.center);
		\draw [in=0, out=180, looseness=0.50] (13.center) to (9.center);
		\draw [in=0, out=180, looseness=0.50] (12.center) to (10.center);
	\end{pgfonlayer}
\end{tikzpicture}

%% file: figures/qpath/axioms/split-aso.tikz
\begin{tikzpicture}
	\begin{pgfonlayer}{nodelayer}
		\node [style=lsplit] (0) at (-3.5, -0.375) {};
		\node [style=none] (1) at (-2.5, -1.125) {};
		\node [style=none] (2) at (-2.5, 0.375) {};
		\node [style=none] (3) at (-4.5, -0.375) {};
		\node [style=none] (4) at (-1, -1.125) {};
		\node [style=lsplit] (5) at (-2, 0.375) {};
		\node [style=none] (6) at (-1, -0.375) {};
		\node [style=none] (7) at (-1, 1.125) {};
		\node [style=none] (8) at (-1, -0.375) {};
		\node [style=none] (9) at (-1, 1.125) {};
		\node [style=lsplit] (10) at (2, 0.375) {};
		\node [style=none] (11) at (3, -0.375) {};
		\node [style=none] (12) at (3, 1.125) {};
		\node [style=none] (13) at (1, 0.375) {};
		\node [style=none] (14) at (4.5, 1.125) {};
		\node [style=lsplit] (15) at (3.5, -0.375) {};
		\node [style=none] (16) at (4.5, -1.125) {};
		\node [style=none] (17) at (4.5, 0.375) {};
		\node [style=none] (18) at (4.5, -1.125) {};
		\node [style=none] (19) at (4.5, 0.375) {};
		\node [style=none] (20) at (0, 0) {$=$};
	\end{pgfonlayer}
	\begin{pgfonlayer}{edgelayer}
		\draw (3.center) to (0);
		\draw [in=180, out=-45, looseness=0.75] (0) to (1.center);
		\draw [in=45, out=180, looseness=0.75] (2.center) to (0);
		\draw (4.center) to (1.center);
		\draw [in=180, out=-45, looseness=0.75] (5) to (6.center);
		\draw [in=45, out=180, looseness=0.75] (7.center) to (5);
		\draw (9.center) to (7.center);
		\draw (8.center) to (6.center);
		\draw (13.center) to (10);
		\draw [in=180, out=-45, looseness=0.75] (10) to (11.center);
		\draw [in=45, out=180, looseness=0.75] (12.center) to (10);
		\draw (14.center) to (12.center);
		\draw [in=180, out=-45, looseness=0.75] (15) to (16.center);
		\draw [in=45, out=180, looseness=0.75] (17.center) to (15);
		\draw (19.center) to (17.center);
		\draw (18.center) to (16.center);
		\draw (5) to (2.center);
		\draw (15) to (11.center);
	\end{pgfonlayer}
\end{tikzpicture}

%% file: figures/qpath/axioms/split-merge-bialgebra.tikz
\begin{tikzpicture}
	\begin{pgfonlayer}{nodelayer}
		\node [style=lsplit] (15) at (-4, 0.75) {};
		\node [style=none] (16) at (-5, -0.75) {};
		\node [style=none] (17) at (-5, 0.75) {};
		\node [style=rmerge] (18) at (-2, 0.75) {};
		\node [style=none] (19) at (0, 0) {$=$};
		\node [style=rmerge] (20) at (-2, -0.75) {};
		\node [style=lsplit] (21) at (-4, -0.75) {};
		\node [style=lsplit] (22) at (3.875, 0) {};
		\node [style=none] (23) at (4.75, -0.75) {};
		\node [style=none] (24) at (4.75, 0.75) {};
		\node [style=rmerge] (25) at (2.125, 0) {};
		\node [style=none] (26) at (1.25, -0.75) {};
		\node [style=none] (27) at (1.25, 0.75) {};
		\node [style=none] (28) at (1, 0.75) {};
		\node [style=none] (29) at (1, -0.75) {};
		\node [style=none] (32) at (-1, -0.75) {};
		\node [style=none] (33) at (-1, 0.75) {};
		\node [style=none] (34) at (5, 0.75) {};
		\node [style=none] (35) at (5, -0.75) {};
	\end{pgfonlayer}
	\begin{pgfonlayer}{edgelayer}
		\draw [style=none] (16.center) to (21);
		\draw [style=none] (17.center) to (15);
		\draw [style=none, bend right=23] (21) to (20);
		\draw [style=none, bend left=23] (15) to (18);
		\draw (18) to (21);
		\draw (20) to (15);
		\draw [in=45, out=180, looseness=0.75] (24.center) to (22);
		\draw [in=180, out=-45, looseness=0.75] (22) to (23.center);
		\draw (26.center) to (29.center);
		\draw (27.center) to (28.center);
		\draw [in=135, out=0, looseness=0.75] (27.center) to (25);
		\draw [in=0, out=-135, looseness=0.75] (25) to (26.center);
		\draw [style=none] (25) to (22);
		\draw (34.center) to (24.center);
		\draw (35.center) to (23.center);
		\draw (32.center) to (20);
		\draw (18) to (33.center);
	\end{pgfonlayer}
\end{tikzpicture}

%% file: figures/qpath/axioms/split-id.tikz
\begin{tikzpicture}
	\begin{pgfonlayer}{nodelayer}
		\node [style=lsplit] (0) at (-2.25, 0) {};
		\node [style=none] (1) at (-3.5, 0) {};
		\node [style=none] (2) at (-1.5, -0.75) {};
		\node [style=none] (3) at (-1.5, 0.75) {};
		\node [style=none] (4) at (-1, -0.75) {};
		\node [style=none] (6) at (0, 0) {$=$};
		\node [style=none] (7) at (1, 0) {};
		\node [style=none] (8) at (2.75, 0) {};
		\node [style=black node] (10) at (-1.25, 0.75) {};
		\node [style=none] (11) at (-1.25, 1.25) {$\scriptstyle 0$};
	\end{pgfonlayer}
	\begin{pgfonlayer}{edgelayer}
		\draw (2.center) to (4.center);
		\draw [in=-45, out=-180, looseness=0.75] (2.center) to (0);
		\draw [in=-180, out=45, looseness=0.75] (0) to (3.center);
		\draw (0) to (1.center);
		\draw (8.center) to (7.center);
		\draw (3.center) to (10);
	\end{pgfonlayer}
\end{tikzpicture}

%% file: zxw-photonics.tex
\begingroup

\allowdisplaybreaks

In this section, we give a method for proving novel equations in $\bf{ZW}_\infty$
by rewriting in the ZXW calculus.
We give a truncated interpretation $\mathcal{T}_d: \bf{ZW}_\infty \to \bf{ZXW}_d$
parametrised by the qudit dimension $d$.
We prove the \emph{lifting} theorem: any equality in the Fock space of the form
$\interp{D} = \interp{D'}$ can be proved by showing that there exists some $N \in \mathbb{N}$
such that $\mathcal{T}_d(D) P_N = \mathcal{T}_d(D') P_N$ for any $d > N$,
where $P_N$ denotes the projector on the $(d - N)$-particle sector of the input qudits.
We use this theorem to give diagrammatic proofs of soundness for the axioms of
the infinite ZW calculus.
%
%

\subsection{Truncation}

\setlength{\abovedisplayskip}{5pt}
\setlength{\belowdisplayskip}{5pt}

Truncation consists in giving a finite dimensional description of operators on an
infinite-dimen\-sional space.
We define this as a mapping $\mathcal{T}_d: \bf{ZW}_\infty \to \bf{ZXW}_d$.
We introduce a useful component in $\bf{ZXW}_d$ representing the projector on
the $d$-particle sector of a pair of qudits:
\begin{equation}
  \label{eq:d-projector}
  \input{qpath/truncated/projector-zxw-lhs.tikz}
  \quad \coloneqq \quad
  \input{qpath/truncated/projector-zxw-rhs.tikz}
  \quad \xmapsto{\interp{\cdot}_d} \quad
  \sum_{\substack{a_1, a_2 = 0 \\ a_1 + a_2  < d}}^{d-1} \ket{a_1, a_2}\bra{a_1, a_2}
\end{equation}
where $\overrightarrow{N} = (1,\,2\,\cdots,\,d - 1)$, $\overrightarrow{-1} = (-1,\,\cdots,\,-1)$, and each operation is defined elementwise.
We can use this to represent bosonic nodes in ZXW.
In fact, the X spider performs addition of basis states modulo $d$, and this agrees
with standard addition whenever the input is in the $d$-particle sector.
We can thus model the truncated split map as an X spider with a projector and
Z boxes for the binomial coefficients.\\
The $d$-truncation $\mathcal{T}_d: \bf{ZW}_\infty \to \bf{ZXW}_d$ is defined on generators
as follows:
\begin{enumerate}
 \item Z spiders and W nodes are mapped to their qudit versions in $\bf{ZXW}_d$.
 \item The split map has the following truncation:
 \begin{equation}
   \label{eq:split_def}
   \input{qpath/generators/split.tikz}
   \quad \xmapsto{\mathcal{T}_d} \quad
   \input{qpath/truncated/split-def.tikz}
   \quad \xmapsto{\interp{\cdot}_d} \quad
   \sum_{n=0}^{d-1}\sum_{k=0}^n \binom{n}{k}^{\frac{1}{2}} \ket{k, n-k}\bra{n}
 \end{equation}
 \item The merge map has the following truncation:
 \begin{equation}
     \label{eq:merge_def}
     \input{qpath/generators/merge.tikz}
     \quad \xmapsto{\mathcal{T}_d} \quad
     \input{qpath/truncated/merge-def.tikz}
     \quad \xmapsto{\interp{\cdot}_d} \quad
     \sum_{\substack{n, \, m = 0  n + m < d}}^{d-1} \binom{n + m}{n}^{\frac{1}{2}} \ket{n+m}\bra{n, m}
 \end{equation}
 \item Particle creations and annihilations correspond to unit X spiders.
 \begin{equation}
   \scalebox{.9}{\input{qpath/truncated/state-def.tikz}}
   \qquad \qquad
   \scalebox{.9}{\input{qpath/truncated/effect-def.tikz}}
   \label{eq:ket-def}
 \end{equation}
\end{enumerate}

This mapping respects all the axioms of the infinite ZW calculus, except for the
bialgebra law between split and merge maps \eqref{rule:bBA}. This rule only holds in $\bf{ZXW}_d$
up to a projector which ensures the total number of particles does not exceed $d$.
That is, in the truncated interpretation we have:
\[
  \scalebox{1}{\input{qpath/truncated/truncated-bialgebra.tikz}}
\]
Starting from the projector on $2$ modes, we may construct the $d$-particle
projector on $m$ modes recursively:
\[
  \scalebox{.9}{\input{qpath/truncated/projectors-def.tikz}}
\]
One may show that the following rewrites are valid in $\bf{ZXW}_d$ for any $d$.
\begin{multicols}{2}
  \setlength{\abovedisplayskip}{0pt}
  \setlength{\belowdisplayskip}{0pt}
  \noindent
  \begin{gather}
    \input{qpath/truncated/projector-square.tikz}
    \tag{PP}\label{rule:PP}\refstepcounter{equation} \\
    \input{qpath/truncated/projector-spider-commute.tikz}
    \tag{PZ}\label{rule:PZ}\refstepcounter{equation} \\ \columnbreak
    \input{qpath/truncated/projector-w-copy.tikz}
    \tag{PW}\label{rule:PW}\refstepcounter{equation} \\
    \input{qpath/truncated/projector-split-copy.tikz}
    \tag{PS}\label{rule:PS}\refstepcounter{equation} \\
    \input{qpath/truncated/projector-state.tikz}
    \tag{PK}\label{rule:PK}\refstepcounter{equation}
  \end{gather}
  \vspace*{\jot}
  \begin{align*}
    \text{where}&\quad
    \vec{1}_j = (\underbrace{1,\dotsc, 1}_{d - j}, 0, \dotsc, 0)
  \end{align*}
\end{multicols}
In fact projectors commute with any map that preserves the number of photons.
Using particle creations and annihilations we may also construct the projector
on the $n$-particle sector for $n < d$ as follows:
\begin{equation}\label{eq:proj-n-def}
  \input{qpath/truncated/projector-n-def.tikz}
\end{equation}

\subsection{Lifting}

We now prove our main result: equalities in the infinite dimensional
calculus may be derived by rewriting in the truncated interpretation.
In order to relate the infinite and truncated interpretations, we use two
ingredients. First, the embedding $\mathcal{E}: \bf{Vect}_d \to \bf{Vect}_\mathbb{N}$
which views any linear map between qudits as a map on the Fock space, acting on the
first $d$ dimensions of each mode. Formally, we have
$\mathcal{E}(f)\ket{x_1, \dots, x_m} = f \ket{x_1, \dots, x_m}$ whenever $x_i < d$
for $i = 1, \dots, m$ and $\mathcal{E}(f)\ket{x_1, \dots, x_m} = 0$ otherwise,
for any linear map $f$ on $m$ qudits. It is easy to show that this embedding is
a faithful monoidal functor.
Second, we use the projector on the $n$-particle sector of the Fock space,
i.e.\@ the linear operator $P_n : \interp{m} \to \interp{m}$ such that
$P_n\ket{x_1, \dots, x_m} = \ket{x_1, \dots, x_m}$ when $\sum_{i=1}^m x_i < n$ and
$P_n\ket{x_1, \dots, x_m} = 0$ otherwise.
The following lemma characterises the infinite interpretation of a diagram in $\bf{ZW}_\infty$ in terms of the truncations.

\begin{lemma}
  \label{lemma:lifting}
  For any diagram $D \in \bf{ZW}_\infty$ and $n \in \mathbb{N}$ there is $d^\ast \in \mathbb{N}$
  such that, whenever $d > d^\ast$, we have:
  \[
    \interp{D} P_n = \mathcal{E}(\interp{\mathcal{T}_d(D)}_d)P_n
  \]
\end{lemma}
\begin{proof}
  To prove the statement it is sufficient that, given $D$ and $n$, we can find a ``big enough'' dimension $d^\ast$,
  such that $\interp{D} \ket{x_1, \dots, x_m} = \mathcal{E}(\interp{\mathcal{T}_{^\ast}(D)}_{d^\ast})\ket{x_1, \dots, x_m}$
  whenever $\sum_{i=1}^m x_i < n$.
  We prove this by induction over the recursive definition of string diagrams.
  Indeed, a diagram in $\bf{ZW}_\infty$ may be viewed as a sequence of layers
  where, in each layer, a single generator appears tensored by identities.
  Therefore, it is sufficient to find $d^\ast$ given an $n$ for each generating layer.
  Note first that W nodes preserve the total number of photons, similarly for
  bosonic nodes, as well as for Z spiders $1 \to 1$ and identities.
  Thus, for layers involving these generators, we can set $d^\ast = n$.
  The Z spiders with $a > 1$ output wires~\eqref{eq:inf-z-spider} can at most multiply the total number
  of photons by $a$, therefore, we can set $d^\ast = a \cdot n$.
  Finite sequences, i.e.\@ Z spider states~\eqref{eq:inf-z-state} only contribute a finite number of photons
  $t$, so we can set $d^\ast = n + t$ in this case.
  Finally, in case of a photon preparation we can set $d^\ast = n + 1$.
  Since the dimension remains finite after each generator, by induction, we can
  find a finite $d^\ast$ for any $D$ and $n$.
\end{proof}

Note that there are potentially smaller values of $d^\ast$ for which the statement above
holds, as well as intermediate dimensions that we can assign to the internal wires of $D$.
Finding optimal values for these dimensions would allow more efficient classical simulation of photonic circuits by tensor network contraction.

\begin{theorem}[Lifting]
  \label{th:lifting}
  For any $D, D': m \to m' \in \bf{ZW}_\infty$ the following are equivalent:
  \begin{enumerate}
    \item In $\bf{Vect}_\mathbb{N}$:
    \[
      \interp{\input{infinite-zw/diagram.tikz}} = \interp{\input{infinite-zw/diagram-prime.tikz}}
    \]
    \item For any $n \in \mathbb{N}$ there is a dimension $d^\ast$ such that,
    for all $d > d^\ast$, in $\bf{ZXW}_d$:
    \[ \input{infinite-zw/diagram-proj.tikz} = \input{infinite-zw/diagram-prime-proj.tikz}
    \]
  \end{enumerate}
\end{theorem}
\begin{proof}
  First note that $\interp{D} = \interp{D'}$ if and only if $\interp{D}P_n = \interp{D'}P_n$
  for any $n \in \mathbb{N}$.
  By \cref{lemma:lifting}, this is equivalent to: for any $n$ there is a $d^\ast$
  such that for all $d > d^\ast$ we have
  $\mathcal{E}(\interp{\mathcal{T}_d(D)}_d)P_n = \mathcal{E}(\interp{\mathcal{T}_d(D')}_d)P_n$.
  Denoting by $\tilde{P}_n$ the projector given in \cref{eq:proj-n-def}, we have
  that $\mathcal{E}(\interp{\mathcal{T}_d(D)\tilde{P}_n}_d) = \mathcal{E}(\interp{\mathcal{T}_d(D')\tilde{P}_n}_d)$
  Finally, using faithfulness of the embedding $\mathcal{E}$ and completeness of the ZXW calculus, we obtain $\mathcal{T}_d(D) \tilde{P}_n  = \mathcal{T}_d(D') \tilde{P}_n$.
\end{proof}

This means in particular that if $\mathcal{T}_d(D) = \mathcal{T}_d(D')$ for any $d \in \mathbb{N}$,
then $\interp{D} = \interp{D'}$.
However, the theorem is stronger than this. To prove $\interp{D} = \interp{D'}$, it
is sufficient to show that there exists $N \in \mathbb{N}$ such that for any $d > N$ we have:
\[
  \input{infinite-zw/zxw-proj.tikz} = \input{infinite-zw/zxw-prime-proj.tikz}
\]
This is proved by setting $d^\ast = n + N$ given any $n$, and is still a weaker statement
as the difference between $d^\ast$ and $n$ may not be constant.
Most of the rules of ZXW are independent of the dimension $d$, making it easier to
rewrite while preserving this condition.

\begin{remark}
    We say that an equation $D = D'$ holds in $\bf{ZW}_\infty$ whenever $\interp{D} = \interp{D'}$ with the standard interpretation $\interp{\cdot} : \bf{ZW}_\infty \to \bf{Vect}_\mathbb{N}$.
    We say that an equation holds in $\bf{ZXW}_d$ if it holds in the rewriting system
    given in \cref{sec:zxw}, and we indicate the rules used.
    Whenever the split or merge map appears in a statement about $\bf{ZXW}_d$,
    we use it as syntactic sugar for its truncation as given in \cref{eq:split_def,eq:merge_def}.
\end{remark}

\begin{restatable}{proposition}{trialgebra}
  \label{prop:trialgebra}
  In $\bf{ZW}_\infty$, the following holds:
  \[
    \input{infinite-zw/trialgebraAZW.tikz}
  \]
\end{restatable}
\begin{proof}
  First, in $\bf{ZXW}_d$ the following holds:
  \begin{align*}
    &\input{infinite-zw/trialgebra-proof-0.tikz} \quad
    \input{infinite-zw/trialgebra-proof-1.tikz} \\
    &\input{infinite-zw/trialgebra-proof-2.tikz}
  \end{align*}
  Precomposing the equation above by projectors gives an equation of the form $\mathcal{T}_d(D)\tilde{P}_d = \mathcal{T}_d(D')\tilde{P}_d$ which holds for any $d$.
  Therefore, by \cref{th:lifting} we deduce that $\interp{D} = \interp{D'}$.
\end{proof}

We may apply the same reasoning to the remaining axioms of $\bf{ZW}_\infty$.

\begin{theorem}[Soundness]
  The axioms of the infinite ZW calculus are sound for the infinite interpretation
  $\interp{\cdot}: \bf{ZW}_\infty \to \bf{Vect}_\mathbb{N}$.
\end{theorem}
\begin{proof}
  The axioms for Z spiders and W nodes, as well as~\eqref{rule:iBsj} and~\eqref{rule:iK0}, are trivially lifted from ZXW since
  they all lie in the image of the truncation functor and
  $\mathcal{T}_d(D) = \mathcal{T}_d(D') \, \forall d \in \mathbb{N} \implies \interp{D} = \interp{D'}$.
  The rules from QPath can be proved with few ZXW rewrites by using the properties of projectors.
  The bialgebra axiom for bosonic nodes~\eqref{rule:bBA} is the hardest one to prove diagrammatically.
  The proof, obtained by checking on the basis elements, boils down to an application
  of the Vandermonde identity, as shown explicitly in~\cite[Section 5.3]{hadzihasanovicAlgebraEntanglementGeometry2017}.
  Proving Rule~\eqref{rule:bW1} diagrammatically requires the definition of
  projector~\eqref{eq:d-projector}; however, it is easier to verify by computing the interpretations. Similarly, Rule~\eqref{rule:iW1} is proved by considering the
  interpretation.
  Rule~\eqref{rule:bTA} is proved in \cref{prop:trialgebra}.
  Rule~\eqref{rule:bZBA} is proved similarly following from the fact that bialgebra between Z and X holds in the truncation.
\end{proof}

Even though dualisers and hadamard nodes are not valid maps in the
infinite calculus, we may still use them when rewriting in the truncation,
as long as they disappear in the resulting diagram. An example is given by \cref{lem:razins},
which is used to prove the commutation relations for bosonic and fermionic operators in the next section.

\endgroup

%% file: hamiltonians.tex
In this section, we use rewriting in ZXW to prove facts about quantum optical Hamiltonians.
These are defined as sums and products of creation and annihilation operators, and may
always be represented as single ZXW diagrams. We give a range of example applications,
from linear optics and non-linear Kerr media to light-matter interaction.
We provide the proof of some propositions and lemmas of this section in the appendix.

\subsection{Sums and products of \texorpdfstring{$\bf{ZW}_\infty$}{ZW∞} diagrams}
Controlled diagrams in $\bf{ZW}_\infty$ are defined as follows.
\begin{definition}[Controlled diagram]
  For a diagram $D$, a controlled diagram $\widetilde{D}$ is
  \begin{equation*}
    \scalebox{1}{\input{hamiltonians/ControlledDiagramDefinition0.tikz}}
    \qquad
    \text{such that}
    \qquad
    \scalebox{1}{\input{hamiltonians/ControlledDiagramDefinition1.tikz}}
    \qquad
    \text{and}
    \qquad
    \scalebox{1}{\input{hamiltonians/ControlledDiagramDefinition2.tikz}}
  \end{equation*}
\end{definition}
As an example, a controlled diagram of identity is given by:
$\quad \scalebox{1}{\input{hamiltonians/IdentityControlledDiagram.tikz}}$
\begin{proposition}[Controlled sum of diagrams]
  \label{SumControlledDiagram}
  Given controlled diagrams $\widetilde{D_1}, \ldots, \widetilde{D_k}$ corresponding to diagrams $D_1, \ldots, D_k$ and complex numbers $c_1, \ldots, c_k$, a controlled diagram for $\sum_i c_i D_i$ is given by
  \begin{equation}
    \scalebox{1}{\input{hamiltonians/ControlledDiagramSum.tikz}}
  \end{equation}
  where $\underline{c_i} = (c_i, 0, \dots, 0)$.
\end{proposition}
\begin{proposition}[Controlled product of diagrams]
  \label{ProductControlledDiagram}
  Given controlled diagrams $\widetilde{D_1}, \ldots, \widetilde{D_k}$ corresponding to the diagrams $D_1, \ldots, D_k$, a controlled diagram for $\prod_i D_i$ is given by
  \begin{equation}
    \scalebox{1}{\input{hamiltonians/ControlledDiagramProduct.tikz}}
  \end{equation}
\end{proposition}
The above propositions can be verified by plugging in $\ket{0}$ and $\ket{1}$.
This allows us to write sums and products of diagrams by creating controlled sums and products,
and then plugging $\ket{1}$ on the top.

\subsection{Creation and annihilation operators}
Hamiltonians are usually written in terms of sums and products of creation/annihilation operators.
In order to obtain them as diagrams in ZXW, we need to build the controlled operators.
Controlled diagrams for bosonic creation ($a^{\dagger}$) and annihilation ($a$) operators are the following:
\begin{equation*}
  \scalebox{1}{\input{hamiltonians/BosonControlledCreation1.tikz}}
  \qquad
  \text{ and }
  \qquad
  \scalebox{1}{\input{hamiltonians/BosonControlledAnnihilation1.tikz}}
\end{equation*}
One may check that plugging $\ket 0$ in the control gives the identity in both cases.
We can now give a diagrammatic proof of the commutation relations by rewriting in ZXW.
\begin{restatable}{proposition}{bosonCommutation}
  In $\bf{ZW}_\infty$, $aa^{\dagger} = a^{\dagger}a + id$, that is:
  \begin{equation}
    \label{eq:boson_commutation}
    \scalebox{1}{\input{hamiltonians/BosonCommutationRelation.tikz}}
  \end{equation}
\end{restatable}

\subsection{Linear optics}
The circuits in linear optics are built using two main gates: (1) phase shift and (2) beam splitter.
These are often defined by their Hamiltonians.
The Hamiltonian of the phase shift is given by the number operator $H_{P} = \alpha \hat{n_1} = \alpha a^{\dagger}_1 a_1$.
The phase shift gate is the exponential of this Hamiltonian, i.e.\@ $e^{i H_P}$.
\begin{restatable}{proposition}{phaseShiftExponential}
  \label{prop:phaseshiftexp}
  In $\bf{ZW}_\infty$, the phase shift gate with phase $\alpha$ is given by:
  \[
    \exp\left(\scalebox{1}{\input{hamiltonians/NumberOperator.tikz}}\right)
    \qquad = \qquad
    \input{hamiltonians/PhaseShiftExponential.tikz}
  \]
\end{restatable}

The Hamiltonian of the beam splitter is given as
$H_{BS} = \theta\left(e^{i\phi} a_1 a^{\dagger}_2 + e^{-i\phi} a^{\dagger}_1 a_2 \right)$
for some parameters $\phi$ and $\theta$. The exponential is given by:
\begin{equation*}
  \exp\left(\input{hamiltonians/examples/BeamSplitter.tikz}\right)
  \qquad = \qquad
  \scalebox{1}{\input{hamiltonians/examples/BeamSplitterExponential.tikz}}
\end{equation*}
where $r=e^{i\phi}\sin \theta$ is the reflectivity, $t=\cos \theta$ is the transmissivity, and $\null^{\ast}$ is complex conjugation.

\subsection{Non-linear optics with Kerr media}

Kerr media are non-linear optical crystals which allow performing entangling operations
between bosonic modes, with applications to photonic quantum computing~\cite{azumaQuantumComputationKerrnonlinear2007}.
Here we study the single mode Kerr effect, described by a phase shift with
a quadratic term, and the cross-Kerr interaction.
We give a representation of the latter as a \emph{phase gadget} which allows
for simple rewrite rules that remove non-linearities from circuits.

The Kerr interaction is given by the Hamiltonian
$H_K = \kappa \hat{n_1}^2 = \kappa a^{\dagger}_1 a_1 a^{\dagger}_1 a_1$.
We can represent this Hamiltonian in ZXW and diagonalise it to get:
\begin{restatable}{proposition}{kerrExample}
  In $\bf{ZW}_\infty$, the Kerr gate with parameter $\kappa$ is given by:
  \[
    \exp\left(\input{hamiltonians/KerrLHS.tikz}\right)
    \qquad = \qquad
    \input{hamiltonians/examples/KerrExponential.tikz}
  \]
\end{restatable}

The cross-Kerr interaction is given by the Hamiltonian
$H_{CK} = \tau \hat{n_1}\hat{n_2} = \tau a^{\dagger}_1 a_1 a^{\dagger}_2 a_2$,
from which we can derive the cross-Kerr gate $CK(\tau) = \exp(i \tau \hat{n}_1 \hat{n}_2)$.
The cross-Kerr gate with $\tau = \pi$ may be used to construct the CZ gate on
dual-rail qubits~\cite{azumaQuantumComputationKerrnonlinear2007}.
We find that it has a natural representation in the calculus.
\begin{restatable}{proposition}{crossKerrExample}
  The cross-Kerr gate with parameter $\kappa$ is given by:
  \begin{equation*}
    \exp\left(\input{hamiltonians/CrossKerrLHS.tikz}\right)
    \qquad = \qquad
    \input{hamiltonians/CrossKerrExp.tikz}
  \end{equation*}
\end{restatable}

Utilizing the rewriting rules of our calculi, we can demonstrate different properties of these operators.
For instance, we can now prove the following proposition related to the composition of two cross-Kerr operators:

\begin{restatable}{proposition}{crossKerrComposition}
  The composition of two cross-Kerr media with parameters $\tau$ and $\mu$ results in a cross-Kerr interaction with parameter $\tau + \mu$.
  Specifically, in $\bf{ZXW}_d$ the following equation holds:
  \begin{equation*}
    \input{hamiltonians/examples/CrossKerrComposition.tikz}
  \end{equation*}
\end{restatable}

\subsection{Towards light-matter interaction}

So far, we have studied only bosonic creation and annihilation operators.
We now introduce \emph{fermionic} creation and annihilation.
These operators are easily represented in $\bf{ZW}_\infty$ using the W node.
The W algebra was indeed already shown to have an important role in fermionic
quantum computing~\cite{ngDiagrammaticCalculusFermionic2019}.
We show the anti-commutation relation for these operators and finally,
we represent the Jaynes-Cummings Hamiltonian describing the interaction of bosons and fermions.

Controlled diagrams for the fermionic creation ($\sigma^+$) and annihilation ($\sigma^-$) operators are given by
\begin{equation*}
  \scalebox{1}{\input{hamiltonians/FermionControlledCreation2.tikz}}
  \qquad
  \text{ and }
  \qquad
  \scalebox{1}{\input{hamiltonians/FermionControlledAnnihilation2.tikz}}
\end{equation*}
Note that the fermionic creation and annihilation operators are only defined on
the qubit subspace due to the Pauli exclusion principle.
These operators acting on the input $\ket{n}$, for $n>2$, result in a zero diagram.

Now, we use the rules of the infinite ZW calculus to prove the anti-commutation relations.
\begin{restatable}{proposition}{fermionCommutation}
  In $\bf{ZW}_\infty$, $\sigma^- \sigma^+ + \sigma^+ \sigma^- = id$, that is:
  \begin{equation}
    \label{eq:fermion_commutation}
    \scalebox{1}{\input{hamiltonians/FermionCommutationRelation.tikz}}
  \end{equation}
  where the right-hand side is the identity on the qubit subspace.
\end{restatable}

The Jaynes-Cummings model describes the interaction between a bosonic mode
and a 2-level atom. It is defined by the following Hamiltonian, for some frequency $\omega$:
\begin{equation}
  H_{JC} = \hbar \omega \left(a_1 \sigma^+_2 + a^{\dagger}_1 \sigma^-_2 \right)
  \qquad \qquad \qquad
  \scalebox{1}{\input{hamiltonians/examples/JaynesCummings.tikz}}
\end{equation}

The Tavis-Cummings Hamiltonian is a natural generalization of Jaynes-Cummings to a setting
in which multiple atoms interact with a bosonic mode:
\begin{equation*}
  H_{TC} = \hbar \left[ \omega_c a_1^{\dagger} a_1
  + \frac{\omega_a}{N} \left(\sum_{n=1}^{N} \sigma^+_n \right) \left(\sum_{n=1}^{N} \sigma^-_n\right)
  + g a_1^{\dagger} \sum_{n=1}^{N} \sigma^-_n + g a_1 \sum_{n=1}^{N} \sigma^+_n \right]
\end{equation*}
where $N$ is the number of atoms, $\omega_a$ and $\omega_c$ are the atomic and cavity
resonance frequency, and $g$ is the atom-photon coupling strength.
\begin{equation*}
  \scalebox{1}{\input{hamiltonians/examples/TavisCummings.tikz}}\label{eq:travis-cummings}
\end{equation*}

%% file: figures/hamiltonians/PhaseShiftExponential.tikz
\begin{tikzpicture}
	\begin{pgfonlayer}{nodelayer}
		\node [style=none] (2) at (-1.5, 0) {};
		\node [style=gbox] (3) at (0, 0) {$e^{i \vec N \alpha}$};
		\node [style=none] (4) at (1.5, 0) {};
	\end{pgfonlayer}
	\begin{pgfonlayer}{edgelayer}
		\draw [in=180, out=0] (2.center) to (4.center);
	\end{pgfonlayer}
\end{tikzpicture}

%% file: figures/hamiltonians/examples/KerrExponential.tikz
\begin{tikzpicture}
	\begin{pgfonlayer}{nodelayer}
		\node [style=none] (4) at (-1.75, 0) {};
		\node [style=gbox] (5) at (0, 0) {$e^{i \kappa \vec{N}^2}$};
		\node [style=none] (6) at (1.75, 0) {};
	\end{pgfonlayer}
	\begin{pgfonlayer}{edgelayer}
		\draw [in=180, out=0] (4.center) to (6.center);
	\end{pgfonlayer}
\end{tikzpicture}

%% file: figures/lemmas/multipliers/multiplier-mod.tikz
\begin{tikzpicture}
	\begin{pgfonlayer}{nodelayer}
		\node [style=none] (0) at (-3.25, 0) {};
		\node [style=none] (1) at (-1, 0) {};
		\node [style=rmat] (2) at (-2.125, 0) {$\minu m$};
		\node [style=none] (3) at (0, 0) {$\coloneqq$};
		\node [style=none] (4) at (1, 0) {};
		\node [style=none] (5) at (4, 0) {};
		\node [style=rmat] (6) at (2.5, 0) {$d \minu m$};
	\end{pgfonlayer}
	\begin{pgfonlayer}{edgelayer}
		\draw (0.center) to (2);
		\draw (2) to (1.center);
		\draw (4.center) to (6);
		\draw (6) to (5.center);
	\end{pgfonlayer}
\end{tikzpicture}

%% file: figures/definitions/hadamard-node.tikz
\begin{tikzpicture}
	\begin{pgfonlayer}{nodelayer}
		\node [style=none] (0) at (1, 0) {};
		\node [style=hbox] (1) at (0, 0) {$H$};
		\node [style=none] (2) at (-1, 0) {};
	\end{pgfonlayer}
	\begin{pgfonlayer}{edgelayer}
		\draw (0.center) to (1);
		\draw (1) to (2.center);
	\end{pgfonlayer}
\end{tikzpicture}

%% file: axioms.tex
\begingroup 

\allowdisplaybreaks
\raggedcolumns

We give a set of rewrite rules that is shown to be complete
in~\cite{poorCompletenessArbitraryFinite2023}.
$\overrightarrow{a} = (a_1,\dotsc, a_{d-1})$ and
$\overrightarrow{b} = (b_1,\dotsc, b_{d-1})$ denote arbitrary complex vectors.

\subsubsection*{Qudit ZX-part of the rules}

\begin{gather}
  \input{axioms/gengspiderfusedit.tikz}
  \tag{S1}\label{rule:S1}\refstepcounter{equation}
\end{gather}
\vspace*{\jot}
\vspace*{-.4cm}
\begin{multicols}{2}
  \setlength{\abovedisplayskip}{0pt}
  \setlength{\belowdisplayskip}{0pt}
  \noindent
  \begin{gather}
    \input{axioms/s2qudit.tikz}
    \tag{S2}\label{rule:S2}\refstepcounter{equation} \\
    \input{axioms/dcomwtha0.tikz}
    \tag{D1}\label{rule:D1}\refstepcounter{equation} \\
    \input{axioms/p1sdit2.tikz}
    \tag{P1}\label{rule:P1}\refstepcounter{equation} \\
    \input{axioms/rdotaemptydit0.tikz}
    \tag{Ept}\label{rule:Ept}\refstepcounter{equation} \\
    \input{axioms/zerotoreddit0.tikz}
    \tag{Zer}\label{rule:Zer}\refstepcounter{equation} \\
    \input{axioms/b2qudit.tikz}
    \tag{B2}\label{rule:B2}\refstepcounter{equation} \\ \pagebreak
    \input{axioms/k1copy.tikz}
    \tag{K0}\label{rule:K0}\refstepcounter{equation} \\
    \input{axioms/pimultiplecpdit.tikz}
    \tag{K1}\label{rule:K1}\refstepcounter{equation} \\
    \input{axioms/k2adit.tikz}
    \tag{K2}\label{rule:K2}\refstepcounter{equation} \\
    \input{axioms/pinkspiders.tikz}
    \tag{HZ}\label{rule:HZ}\refstepcounter{equation}
  \end{gather}
\end{multicols}
\vspace{-0.75cm}
\begin{equation*}
  \text{where}\quad
  \protect\overleftarrow{a}=(a_{d-1}, \dotsc, a_1),\quad
  \protect\overrightarrow{ab}=(a_1 b_1,\dotsc, a_{d-1} b_{d-1}),
  \quad\text{and}\quad
  \protect{k_j(\overrightarrow{a})}=\left(\frac{a_{1-j}}{a_{d-j}}, \dotsc, \frac{a_{d-1-j}}{a_{d-j}}\right)
\end{equation*}
In the definition of $k_j$ we take the indices modulo $d$, that is, $a_j = a_{(j\ \mathrm{mod}\ d)}$.

\subsubsection*{Qudit ZW-part of the rules}
\begin{multicols}{2}
  \setlength{\abovedisplayskip}{0pt}
  \setlength{\belowdisplayskip}{0pt}
  \noindent
  \begin{gather}
    \input{axioms/w-bialgebra.tikz}
    \tag{BZW}\label{rule:BZW}\refstepcounter{equation} \\
    \input{axioms/phasecopydit.tikz}
    \tag{Pcy}\label{rule:Pcy}\refstepcounter{equation} \\
    \input{axioms/additiondit.tikz}
    \tag{AD}\label{rule:AD}\refstepcounter{equation} \\ \columnbreak
    \input{axioms/wsymetrydit.tikz}
    \tag{Sym}\label{rule:Sym}\refstepcounter{equation} \\
    \input{axioms/associatedit.tikz}
    \tag{Aso}\label{rule:Aso}\refstepcounter{equation} \\
    \input{axioms/w-w-algebra.tikz}
    \tag{WW}\label{rule:WW}\refstepcounter{equation}
  \end{gather}
\end{multicols}

\subsubsection*{Qudit ZXW-part of the rules}

\begin{multicols}{2}
  \setlength{\abovedisplayskip}{0pt}
  \setlength{\belowdisplayskip}{0pt}
  \noindent
  \begin{gather}
    \input{axioms/triangleocopydit.tikz}
    \tag{Bs0}\label{rule:Bs0}\refstepcounter{equation} \\
    \input{axioms/trianglepicopydit2.tikz}
    \tag{Bsj}\label{rule:Bsj}\refstepcounter{equation} \\
    \input{axioms/trialgebra.tikz}
    \tag{TA}\label{rule:TA}\refstepcounter{equation} \\
    \input{axioms/v-push-w.tikz}
    \tag{VW}\label{rule:VW}\refstepcounter{equation} \\
    \input{axioms/z-push-v.tikz}
    \tag{ZV}\label{rule:ZV}\refstepcounter{equation} \\\columnbreak
    \input{axioms/zbox-v-decomposition-2.tikz}
    \tag{VA}\label{rule:VA}\refstepcounter{equation} \\
    \input{axioms/k1zstate.tikz}
    \tag{KZ}\label{rule:KZ}\refstepcounter{equation}
  \end{gather}
\end{multicols}
\vspace{-0.75cm}
\begin{align*}
  \text{where}&\quad
  T_j=\overbrace{(\underbrace{0,\dotsc, 1}_{d-j}, \dotsc, 0)}^{d-1},\quad
  e_1, \dotsc, e_n \in \{ 1, \dotsc, d-1\},\quad
  \overrightarrow{a_{d-1}} = \left(a_{d-1}, a_{d-1}, \dotsc, a_{d-1}\right)
\end{align*}

\endgroup 

%% file: figures/axioms/zerotoreddit0.tikz
\begin{tikzpicture}
	\begin{pgfonlayer}{nodelayer}
		\node [style=gbox] (0) at (-2.25, 0) {$0$};
		\node [style=rn] (1) at (1.25, 0) {};
		\node [style=none] (2) at (-1, 0) {};
		\node [style=none] (3) at (2.5, 0) {};
		\node [style=none] (4) at (0, 0) {=};
	\end{pgfonlayer}
	\begin{pgfonlayer}{edgelayer}
		\draw (2.center) to (0);
		\draw (1) to (3.center);
	\end{pgfonlayer}
\end{tikzpicture}

%% file: app-lemmas.tex
\setlength{\abovedisplayskip}{10pt}
\setlength{\belowdisplayskip}{10pt}
\allowdisplaybreaks
\setlength{\jot}{20pt}

\section{Proof of lemmas}

\begin{lemma}
  \label{s4lm}\cite{wangQufiniteZXcalculusUnified2022}
  In $\bf{ZXW}_d$, for any $d$, we have:
  \begin{gather}
    \input{zxw/redspider0pfusedit2.tikz}
    \tag{S4}\label{rule:S4}
  \end{gather}
\end{lemma}
\begin{proof}
  Same as \zxw{Lemma 13}.
\end{proof}

\begin{lemma}
  \label{xdualiserlm}
  In $\bf{ZXW}_d$, for any $d$, we have:
  \begin{align*}
    &\input{zxw/xdualiser-1.tikz}\\
    &\input{zxw/xdualiser-2.tikz}
  \end{align*}
\end{lemma}
\begin{proof}
  Same as \zxw{Lemma 25}.
\end{proof}

\begin{lemma}
  \label{dboxgcopylm}\cite{wangQufiniteZXcalculusUnified2022}
  In $\bf{ZXW}_d$, for any $d$, we have:
  \[
    \input{zxw/dboxgcopy.tikz}
  \]
\end{lemma}
\begin{proof}
  Same as \zxw{Lemma 35}.
\end{proof}

\begin{lemma}
  \label{lem:square-popping}
  In $\bf{ZXW}_d$, for any $d$, we have:
  \[
    \input{lemmas/SquarePopping.tikz}
  \]
\end{lemma}
\begin{proof}
  \[
    \input{lemmas/SquarePoppingProof.tikz}
  \]
\end{proof}

\begin{lemma}
  In $\bf{ZW}_\infty$, we have:
  \label{lem:wunit}
  \[
    \input{lemmas/wunit.tikz}
  \]
\end{lemma}
\begin{proof}
  The proof is obtained by lifting \zxw{Lemma 38}.
\end{proof}

\begin{lemma}
  \label{lem:state-phase}
  In $\bf{ZW}_\infty$ we have:
  \[
    \input{lemmas/state-phase.tikz}
  \]
\end{lemma}
\begin{proof}
  Firstly, we show that in $\bf{ZXW}_d$ we have:
  \[
    \input{lemmas/state-phase-proof.tikz}
  \]
  Then, the proof is acquired by lifting.
\end{proof}

\begin{lemma}
  \label{lem:split-transpose}
  In $\bf{ZW}_\infty$,
  \[
    \input{lemmas/MergeTranspose.tikz}
    \qquad\qquad
    \input{lemmas/WTranspose.tikz}
  \]
\end{lemma}
\begin{proof}
  The statement follows from the lifting theorem as it is a trivial equation in $\bf{ZXW}_d$.
\end{proof}

\begin{lemma}
  \label{lem:razins}
  In $\bf{ZW}_\infty$ we have:
  \[
    \input{lemmas/RazinsLemma.tikz}
  \]
\end{lemma}
\begin{proof}
  The first equation follows from Axioms~\eqref{rule:iK0} and~\eqref{rule:bW1}.
  For the second equation, we start by truncating and reasoning in $\bf{ZXW}_d$:
  \begin{align*}
    \input{lemmas/RazinsLemmaProof0.tikz} \quad
    &\input{lemmas/RazinsLemmaProof1.tikz} \\
    &\input{lemmas/RazinsLemmaProof2.tikz} \\
    &\input{lemmas/RazinsLemmaProof4.tikz} \\
    &\input{lemmas/RazinsLemmaProof5.tikz}
  \end{align*}

  We lift the derivation above using \cref{th:lifting} and end by reasoning in $\bf{ZW}_\infty$:
  \begin{align*}
    \input{lemmas/RazinsLemmaProof0-2.tikz} \quad
    &\input{lemmas/RazinsLemmaProof6.tikz} \\
     &\input{lemmas/RazinsLemmaProof8.tikz}
  \end{align*}
\end{proof}


\begin{restatable}{lemma}{phaseHamiltonian}
  \label{lem:phaseHamiltonian}
  \begin{equation*}
    \scalebox{1}{\input{hamiltonians/NumberOperator.tikz}}
    \qquad \xmapsto{T_d} \qquad
    \scalebox{1}{\input{hamiltonians/NumberOperatorRHS.tikz}}
  \end{equation*}
\end{restatable}
\begin{proof}
    \begin{align*}
      \input{hamiltonians/NumberOperatorProof-0.tikz} \quad
      &\input{hamiltonians/NumberOperatorProof-1.tikz} \\
      &\input{hamiltonians/NumberOperatorProof-2.tikz} \\
      &\input{hamiltonians/NumberOperatorProof-3.tikz} \\
      &\input{hamiltonians/NumberOperatorProof-4.tikz} \\
      &\input{hamiltonians/NumberOperatorProof-5.tikz} \\
      &\input{hamiltonians/NumberOperatorProof-6.tikz} \\
      &\input{hamiltonians/NumberOperatorProof-7.tikz}
    \end{align*}
\end{proof}

\begin{lemma}
  \label{lem:phaseShiftExponential}
  In $\bf{ZXW}_d$,
  \[
    \exp\left(\input{hamiltonians/NumberOperatorRHS.tikz}\right)
    \qquad = \qquad
    \input{hamiltonians/PhaseShiftExponential.tikz}
  \]
\end{lemma}
\begin{proof}
  \begin{align*}
    &\input{hamiltonians/PhaseShiftExponentialProof-1.tikz} \\
    &\input{hamiltonians/PhaseShiftExponentialProof-2.tikz}
  \end{align*}
\end{proof}

\bosonCommutation*
\begin{proof}
  First, by rewriting in $\bf{ZW}_\infty$, and then using \cref{lem:razins} lifted from $\bf{ZXW}_d$.
  \begin{align*}
    \input{hamiltonians/BosonCommutationRelationProof-0.tikz} \quad
    & \input{hamiltonians/BosonCommutationRelationProof-1.tikz} \\
    & \input{hamiltonians/BosonCommutationRelationProof-2.tikz}
  \end{align*}
\end{proof}

\phaseShiftExponential*
\begin{proof}
  First, we diagonalise the truncation of the Hamiltonian $H_{P}$ using ZXW rewrites:
  \begin{equation*}
    \scalebox{1}{\input{hamiltonians/NumberOperator.tikz}}
    \qquad \xmapsto{\mathcal{T}_d} \qquad
    \scalebox{1}{\input{hamiltonians/NumberOperatorRHS.tikz}}
  \end{equation*}
  (See \cref{lem:phaseHamiltonian})
  Then, we derive the exponential in ZXW:
  \[
    \exp\left(\input{hamiltonians/NumberOperatorRHS.tikz}\right)
    \qquad = \qquad
    \input{hamiltonians/PhaseShiftExponential.tikz}
  \]
  (See \cref{lem:phaseShiftExponential})
  Lastly, the result is obtained by lifting.
\end{proof}

\kerrExample*
\begin{proof}
  In $\bf{ZXW}_d$
  \begin{align*}
    &\input{hamiltonians/examples/KerrProof-1.tikz}\\
    &\input{hamiltonians/examples/KerrProof-2.tikz}
  \end{align*}
  Then, one may exponentiate the diagram in the same way as \cref{lem:phaseShiftExponential}.
\end{proof}

\crossKerrExample*
\begin{proof}
  First, we demonstrate the following relation in $\bf{ZXW}_d$:
  \[
    \input{hamiltonians/CrossKerrProof-1.tikz}
  \]
  Based on this form, it becomes evident that the diagram corresponds to the operator $i \tau \hat{n}_1 \hat{n}_2$, which needs to be exponentiated.
  We employ the technique described by van de Wetering and Yeh~\cite{vandeweteringPhaseGadgetCompilation2022}
  for constructing diagonal qudit gates using \emph{phase gadgets}~\cite{cowtanPhaseGadgetSynthesis2019,kissingerReducingNumberNonClifford2020}.
  Firstly, note that the polynomial equation $(x + y)^2 = x^2 + 2xy + y^2$ implies that
  $e^{i \tau x y} = e^{i \frac{\tau}{2} \left( (x + y)^2 - x^2 - y^2 \right)} =
  e^{i \frac{\tau}{2} (x + y)^2} e^{-i \frac{\tau}{2} x^2} e^{-i \frac{\tau}{2} y^2}$.
  The same equation holds when considering $x$ and $y$ as operators $\hat{n}_1$ and $\hat{n}_2$, respectively, given their commutation.
  Utilizing these calculations, we can express the truncated cross-Kerr gate as the composition of the number operator on each wire and the diagram corresponding to $e^{i \frac{\tau}{2} (x + y)^2}$.
  The former is given by \cref{prop:phaseshiftexp}, while the latter is represented by the diagram in $\bf{ZXW}_d$:
  \[
    \input{hamiltonians/CrossKerrPart.tikz}
  \]
  By composing the diagrams and sliding the projector through the green phases according to \eqref{rule:PZ}, we obtain the cross-Kerr operator in $\bf{ZXW}_d$:
  \[
    \input{definitions/cross-kerr.tikz}
  \]
  Lifting this diagram to $\bf{ZW}_\infty$ completes the proof of the proposition.
\end{proof}

\crossKerrComposition*
\begin{proof}
  In $\bf{ZXW}_d$, the following holds:
  \begin{align*}
    &\input{hamiltonians/examples/CrossKerrCompositionProof-1.tikz} \\
    &\input{hamiltonians/examples/CrossKerrCompositionProof-2.tikz} \\
    &\input{hamiltonians/examples/CrossKerrCompositionProof-3.tikz} \\
    &\input{hamiltonians/examples/CrossKerrCompositionProof-4.tikz}
  \end{align*}
\end{proof}

\fermionCommutation*
\begin{proof}
    We reason entirely in $\bf{ZW}_\infty$:
  \begin{align*}
    &\input{hamiltonians/FermionCommutationRelationProof-1.tikz} \\
    &\input{hamiltonians/FermionCommutationRelationProof-2.tikz} \\
    &\input{hamiltonians/FermionCommutationRelationProof-3.tikz} \\
    &\input{hamiltonians/FermionCommutationRelationProof-4.tikz}
  \end{align*}
\end{proof}

%% file: figures/zxw/redspider0pfusedit2.tikz
\begin{tikzpicture}
	\begin{pgfonlayer}{nodelayer}
		\node [style=none] (17) at (3.75, 0) {$=$};
		\node [style=none] (48) at (-4, 0) {$=$};
		\node [style=none] (77) at (-10.125, 2.375) {};
		\node [style=none] (78) at (-8.375, 2.375) {};
		\node [style=none] (81) at (-9.25, 2.875) {\scriptsize $n$};
		\node [style={rn_phase}] (141) at (-6.25, 0.75) {$K_j$ };
		\node [style=none] (142) at (-5.5, 1.75) {};
		\node [style=none] (143) at (-7, 1.75) {};
		\node [style=none] (144) at (-5.5, -1.5) {};
		\node [style=none] (145) at (-7, -1.5) {};
		\node [style=none] (146) at (-6.25, -1.75) {$\cdots$};
		\node [style=none] (147) at (-6.25, 1.75) {$\cdots$};
		\node [style={rn_phase}] (148) at (-9.25, -0.75) {$K_i$ };
		\node [style=none] (149) at (-8.5, 1.5) {};
		\node [style=none] (150) at (-10, 1.5) {};
		\node [style=none] (151) at (-8.5, -1.75) {};
		\node [style=none] (152) at (-10, -1.75) {};
		\node [style=none] (153) at (-9.25, -1.75) {$\cdots$};
		\node [style=none] (154) at (-9.25, 1.75) {$\cdots$};
		\node [style={rn_phase}] (155) at (7.5, 0) {$K_{i + j}$};
		\node [style=none] (156) at (9.75, -1.5) {};
		\node [style=none] (157) at (8.25, -1.5) {};
		\node [style=none] (158) at (6.75, -1.5) {};
		\node [style=none] (159) at (5.25, -1.5) {};
		\node [style=none] (160) at (9.75, 1.5) {};
		\node [style=none] (161) at (8.25, 1.5) {};
		\node [style=none] (162) at (6.75, 1.5) {};
		\node [style=none] (163) at (5.25, 1.5) {};
		\node [style=none] (168) at (-8.5, 2) {};
		\node [style=none] (169) at (-10, 2) {};
		\node [style=none] (170) at (-5.5, -2) {};
		\node [style=none] (171) at (-7, -2) {};
		\node [style=none] (172) at (9, -1.875) {$\cdots$};
		\node [style=none] (173) at (9, 1.875) {$\cdots$};
		\node [style=none] (174) at (6, -1.875) {$\cdots$};
		\node [style=none] (175) at (6, 1.875) {$\cdots$};
		\node [style=none] (176) at (-7, 2) {};
		\node [style=none] (177) at (-5.5, 2) {};
		\node [style=none] (178) at (-8.5, -2) {};
		\node [style=none] (179) at (-10, -2) {};
		\node [style=none] (180) at (-7.125, 2.375) {};
		\node [style=none] (181) at (-5.375, 2.375) {};
		\node [style=none] (182) at (-6.25, 2.875) {\scriptsize $s$};
		\node [style=none] (183) at (-8.375, -2.375) {};
		\node [style=none] (184) at (-10.125, -2.375) {};
		\node [style=none] (185) at (-9.25, -2.875) {\scriptsize $m$};
		\node [style=none] (186) at (-5.375, -2.375) {};
		\node [style=none] (187) at (-7.125, -2.375) {};
		\node [style=none] (188) at (-6.25, -2.875) {\scriptsize $t$};
		\node [style=none] (189) at (-2.5, 2.375) {};
		\node [style=none] (190) at (-0.75, 2.375) {};
		\node [style=none] (191) at (-1.625, 2.875) {\scriptsize $n$};
		\node [style={rn_phase}] (192) at (1.375, -0.75) {$K_j$ };
		\node [style=none] (193) at (2.125, 1.5) {};
		\node [style=none] (194) at (0.625, 1.5) {};
		\node [style=none] (195) at (2.125, -1.75) {};
		\node [style=none] (196) at (0.625, -1.75) {};
		\node [style=none] (197) at (1.375, -1.75) {$\cdots$};
		\node [style=none] (198) at (1.375, 1.75) {$\cdots$};
		\node [style={rn_phase}] (199) at (-1.625, 0.75) {$K_i$ };
		\node [style=none] (200) at (-0.875, 1.75) {};
		\node [style=none] (201) at (-2.375, 1.75) {};
		\node [style=none] (202) at (-0.875, -1.5) {};
		\node [style=none] (203) at (-2.375, -1.5) {};
		\node [style=none] (204) at (-1.625, -1.75) {$\cdots$};
		\node [style=none] (205) at (-1.625, 1.75) {$\cdots$};
		\node [style=none] (206) at (-0.875, 2) {};
		\node [style=none] (207) at (-2.375, 2) {};
		\node [style=none] (208) at (2.125, -2) {};
		\node [style=none] (209) at (0.625, -2) {};
		\node [style=none] (210) at (0.625, 2) {};
		\node [style=none] (211) at (2.125, 2) {};
		\node [style=none] (212) at (-0.875, -2) {};
		\node [style=none] (213) at (-2.375, -2) {};
		\node [style=none] (214) at (0.5, 2.375) {};
		\node [style=none] (215) at (2.25, 2.375) {};
		\node [style=none] (216) at (1.375, 2.875) {\scriptsize $s$};
		\node [style=none] (217) at (-0.75, -2.375) {};
		\node [style=none] (218) at (-2.5, -2.375) {};
		\node [style=none] (219) at (-1.625, -2.875) {\scriptsize $m$};
		\node [style=none] (220) at (2.25, -2.375) {};
		\node [style=none] (221) at (0.5, -2.375) {};
		\node [style=none] (222) at (1.375, -2.875) {\scriptsize $t$};
		\node [style=none] (244) at (11.25, 0) {$=$};
		\node [style=none] (274) at (5.25, 2) {};
		\node [style=none] (275) at (6.75, 2) {};
		\node [style=none] (276) at (8.25, 2) {};
		\node [style=none] (277) at (9.75, 2) {};
		\node [style=none] (278) at (5.25, -2) {};
		\node [style=none] (279) at (6.75, -2) {};
		\node [style=none] (280) at (8.25, -2) {};
		\node [style=none] (281) at (9.75, -2) {};
		\node [style=none] (284) at (7.45, 2.925) {\scriptsize $n + s$};
		\node [style=none] (285) at (5.125, 2.375) {};
		\node [style=none] (286) at (9.875, 2.375) {};
		\node [style=none] (291) at (9.875, -2.375) {};
		\node [style=none] (292) at (5.125, -2.375) {};
		\node [style=none] (293) at (7.35, -2.925) {\scriptsize $m + t$};
		\node [style={rn_phase}] (294) at (15.375, 0) {$K_{i + j}$};
		\node [style=none] (295) at (17.625, -1.5) {};
		\node [style=none] (296) at (16.125, -1.75) {};
		\node [style=none] (297) at (14.625, -1.75) {};
		\node [style=none] (298) at (13.125, -1.5) {};
		\node [style=none] (299) at (17.625, 1.5) {};
		\node [style=none] (300) at (16.125, 1.75) {};
		\node [style=none] (301) at (14.625, 1.75) {};
		\node [style=none] (302) at (13.125, 1.5) {};
		\node [style=none] (303) at (16.875, -1.875) {$\cdots$};
		\node [style=none] (304) at (16.875, 1.875) {$\cdots$};
		\node [style=none] (305) at (13.875, -1.875) {$\cdots$};
		\node [style=none] (306) at (13.875, 1.875) {$\cdots$};
		\node [style=none] (307) at (13.125, 2) {};
		\node [style=none] (308) at (14.625, 2) {};
		\node [style=none] (309) at (16.125, 2) {};
		\node [style=none] (310) at (17.625, 2) {};
		\node [style=none] (311) at (13.125, -2) {};
		\node [style=none] (312) at (14.625, -2) {};
		\node [style=none] (313) at (16.125, -2) {};
		\node [style=none] (314) at (17.625, -2) {};
		\node [style=none] (315) at (15.325, 2.925) {\scriptsize $n + s$};
		\node [style=none] (316) at (13, 2.375) {};
		\node [style=none] (317) at (17.75, 2.375) {};
		\node [style=none] (318) at (17.75, -2.375) {};
		\node [style=none] (319) at (13, -2.375) {};
		\node [style=none] (320) at (15.225, -2.925) {\scriptsize $m + t$};
		\node [style=none] (321) at (14.625, 1) {};
		\node [style=none] (322) at (16.125, 1) {};
		\node [style=none] (323) at (14.625, -1) {};
		\node [style=none] (324) at (16.125, -1) {};
		\node [style=none] (325) at (15.45, 1.875) {$\cdots$};
		\node [style=none] (326) at (15.45, -1.9) {$\cdots$};
	\end{pgfonlayer}
	\begin{pgfonlayer}{edgelayer}
		\draw [style=braceedge] (77.center) to (78.center);
		\draw [in=90, out=-60] (141) to (144.center);
		\draw [in=90, out=-120] (141) to (145.center);
		\draw [in=270, out=45] (141) to (142.center);
		\draw [in=270, out=135] (141) to (143.center);
		\draw [in=90, out=-45] (148) to (151.center);
		\draw [in=90, out=-135] (148) to (152.center);
		\draw [in=270, out=60] (148) to (149.center);
		\draw [in=270, out=120] (148) to (150.center);
		\draw [in=90, out=-60] (155) to (157.center);
		\draw [in=90, out=-30] (155) to (156.center);
		\draw [in=90, out=-120] (155) to (158.center);
		\draw [in=90, out=-150] (155) to (159.center);
		\draw [in=-90, out=30] (155) to (160.center);
		\draw [in=-90, out=60] (155) to (161.center);
		\draw [in=-90, out=120] (155) to (162.center);
		\draw [in=-90, out=150] (155) to (163.center);
		\draw (168.center) to (149.center);
		\draw (169.center) to (150.center);
		\draw (145.center) to (171.center);
		\draw (144.center) to (170.center);
		\draw [in=45, out=-135, looseness=0.75] (141) to (148);
		\draw (177.center) to (142.center);
		\draw (176.center) to (143.center);
		\draw (179.center) to (152.center);
		\draw (151.center) to (178.center);
		\draw [style=braceedge] (180.center) to (181.center);
		\draw [style=braceedge] (183.center) to (184.center);
		\draw [style=braceedge] (186.center) to (187.center);
		\draw [style=braceedge] (189.center) to (190.center);
		\draw [in=90, out=-45] (192) to (195.center);
		\draw [in=90, out=-135] (192) to (196.center);
		\draw [in=270, out=60] (192) to (193.center);
		\draw [in=270, out=120] (192) to (194.center);
		\draw [in=90, out=-60] (199) to (202.center);
		\draw [in=90, out=-120] (199) to (203.center);
		\draw [in=270, out=45] (199) to (200.center);
		\draw [in=270, out=135] (199) to (201.center);
		\draw (206.center) to (200.center);
		\draw (207.center) to (201.center);
		\draw (196.center) to (209.center);
		\draw (195.center) to (208.center);
		\draw [in=-45, out=135, looseness=0.75] (192) to (199);
		\draw (211.center) to (193.center);
		\draw (210.center) to (194.center);
		\draw (213.center) to (203.center);
		\draw (202.center) to (212.center);
		\draw [style=braceedge] (214.center) to (215.center);
		\draw [style=braceedge] (217.center) to (218.center);
		\draw [style=braceedge] (220.center) to (221.center);
		\draw (277.center) to (160.center);
		\draw (276.center) to (161.center);
		\draw (275.center) to (162.center);
		\draw (274.center) to (163.center);
		\draw (281.center) to (156.center);
		\draw (157.center) to (280.center);
		\draw (158.center) to (279.center);
		\draw (159.center) to (278.center);
		\draw [style=braceedge] (285.center) to (286.center);
		\draw [style=braceedge] (291.center) to (292.center);
		\draw [in=90, out=-30] (294) to (295.center);
		\draw [in=90, out=-150] (294) to (298.center);
		\draw [in=-90, out=30] (294) to (299.center);
		\draw [in=-90, out=150] (294) to (302.center);
		\draw (310.center) to (299.center);
		\draw (309.center) to (300.center);
		\draw (308.center) to (301.center);
		\draw (307.center) to (302.center);
		\draw (314.center) to (295.center);
		\draw (296.center) to (313.center);
		\draw (297.center) to (312.center);
		\draw (298.center) to (311.center);
		\draw [style=braceedge] (316.center) to (317.center);
		\draw [style=braceedge] (318.center) to (319.center);
		\draw [in=90, out=-90, looseness=0.50] (301.center) to (322.center);
		\draw [in=90, out=-90, looseness=0.50] (300.center) to (321.center);
		\draw [in=120, out=-90, looseness=0.75] (321.center) to (294);
		\draw [in=60, out=-90, looseness=0.75] (322.center) to (294);
		\draw [in=240, out=90, looseness=0.75] (323.center) to (294);
		\draw [in=90, out=-60, looseness=0.75] (294) to (324.center);
		\draw [in=90, out=-90, looseness=0.50] (324.center) to (297.center);
		\draw [in=90, out=-90, looseness=0.50] (323.center) to (296.center);
	\end{pgfonlayer}
\end{tikzpicture}

%% file: figures/zxw/xdualiser-1.tikz
\begin{tikzpicture}
	\begin{pgfonlayer}{nodelayer}
		\node [style=none] (1) at (-11, 1.25) {};
		\node [style=none] (2) at (-10, 1.25) {};
		\node [style=none] (3) at (-10.5, 1) {$\cdots$};
		\node [style=rn] (4) at (-10.5, 0) {};
		\node [style=none] (5) at (-11, -1.25) {};
		\node [style=none] (6) at (-10, -1.25) {};
		\node [style=none] (7) at (-10.5, -1) {$\cdots$};
		\node [style=none] (8) at (-11.75, -1.25) {};
		\node [style=none] (9) at (-9, 0) {$=$};
		\node [style=none] (10) at (-7, 1.25) {};
		\node [style=none] (11) at (-6, 1.25) {};
		\node [style=none] (12) at (-6.5, 1) {$\cdots$};
		\node [style=rn] (13) at (-6.5, 0) {};
		\node [style=none] (14) at (-7, -1.25) {};
		\node [style=none] (15) at (-6, -1.25) {};
		\node [style=none] (16) at (-6.5, -1) {$\cdots$};
		\node [style=none] (17) at (-7.75, -1.25) {};
		\node [style=none] (18) at (-7.375, 0.75) {};
		\node [style=hbox] (19) at (-7.75, -0.25) {$D$};
		\node [style=none] (20) at (2.25, 1.25) {};
		\node [style=none] (21) at (3.25, 1.25) {};
		\node [style=none] (22) at (2.75, 1) {$\cdots$};
		\node [style=rn] (23) at (2.75, 0) {};
		\node [style=none] (24) at (2.25, -1.25) {};
		\node [style=none] (25) at (3.25, -1.25) {};
		\node [style=none] (26) at (2.75, -1) {$\cdots$};
		\node [style=none] (27) at (1.5, -1.25) {};
		\node [style=none] (28) at (4.25, 0) {$=$};
		\node [style=none] (29) at (6.25, 1.25) {};
		\node [style=none] (30) at (7.25, 1.25) {};
		\node [style=none] (31) at (6.75, 1) {$\cdots$};
		\node [style=rn] (32) at (6.75, 0) {};
		\node [style=none] (33) at (6.25, -1.25) {};
		\node [style=none] (34) at (7.25, -1.25) {};
		\node [style=none] (35) at (6.75, -1) {$\cdots$};
		\node [style=none] (36) at (5.5, -1.25) {};
		\node [style=none] (37) at (5.875, 0.75) {};
		\node [style=hbox] (38) at (1.85, -0.6) {$D$};
		\node [style=none] (77) at (-4.75, 0) {$=$};
		\node [style=none] (78) at (-2.5, 1.25) {};
		\node [style=none] (79) at (-1.5, 1.25) {};
		\node [style=none] (80) at (-2, 1) {$\cdots$};
		\node [style=rn] (81) at (-2, 0) {};
		\node [style=none] (82) at (-2.5, -1.25) {};
		\node [style=none] (83) at (-1.5, -1.25) {};
		\node [style=none] (84) at (-2, -1) {$\cdots$};
		\node [style=none] (85) at (-3.5, -1.25) {};
		\node [style=none] (86) at (-3.375, 1) {};
		\node [style=hbox] (87) at (-2.875, 0.5) {$D$};
	\end{pgfonlayer}
	\begin{pgfonlayer}{edgelayer}
		\draw [in=-45, out=90] (6.center) to (4);
		\draw [in=-135, out=90] (5.center) to (4);
		\draw [in=45, out=-90] (2.center) to (4);
		\draw [in=135, out=-90] (1.center) to (4);
		\draw [in=-150, out=90] (8.center) to (4);
		\draw [in=-45, out=90] (15.center) to (13);
		\draw [in=-135, out=90] (14.center) to (13);
		\draw [in=45, out=-90] (11.center) to (13);
		\draw [in=135, out=-90] (10.center) to (13);
		\draw [in=150, out=0, looseness=0.50] (18.center) to (13);
		\draw [in=90, out=-180, looseness=0.50] (18.center) to (17.center);
		\draw [in=-45, out=90] (25.center) to (23);
		\draw [in=-135, out=90] (24.center) to (23);
		\draw [in=45, out=-90] (21.center) to (23);
		\draw [in=135, out=-90] (20.center) to (23);
		\draw [in=-45, out=90] (34.center) to (32);
		\draw [in=-135, out=90] (33.center) to (32);
		\draw [in=45, out=-90] (30.center) to (32);
		\draw [in=135, out=-90] (29.center) to (32);
		\draw [in=150, out=0, looseness=0.50] (37.center) to (32);
		\draw [in=90, out=-180, looseness=0.50] (37.center) to (36.center);
		\draw [in=90, out=-150] (23) to (27.center);
		\draw [in=-45, out=90] (83.center) to (81);
		\draw [in=-135, out=90] (82.center) to (81);
		\draw [in=45, out=-90] (79.center) to (81);
		\draw [in=135, out=-90] (78.center) to (81);
		\draw [in=180, out=0, looseness=0.75] (86.center) to (81);
		\draw [in=90, out=-180, looseness=0.50] (86.center) to (85.center);
	\end{pgfonlayer}
\end{tikzpicture}

%% file: figures/zxw/xdualiser-2.tikz
\begin{tikzpicture}
	\begin{pgfonlayer}{nodelayer}
		\node [style=none] (0) at (2.25, -1.25) {};
		\node [style=none] (1) at (3.25, -1.25) {};
		\node [style=none] (2) at (2.75, -1) {$\cdots$};
		\node [style=rn] (3) at (2.75, 0) {};
		\node [style=none] (4) at (2.25, 1.25) {};
		\node [style=none] (5) at (3.25, 1.25) {};
		\node [style=none] (6) at (2.75, 1) {$\cdots$};
		\node [style=none] (7) at (1.5, 1.25) {};
		\node [style=none] (8) at (4.25, 0) {$=$};
		\node [style=none] (9) at (6.25, -1.25) {};
		\node [style=none] (10) at (7.25, -1.25) {};
		\node [style=none] (11) at (6.75, -1) {$\cdots$};
		\node [style=rn] (12) at (6.75, 0) {};
		\node [style=none] (13) at (6.25, 1.25) {};
		\node [style=none] (14) at (7.25, 1.25) {};
		\node [style=none] (15) at (6.75, 1) {$\cdots$};
		\node [style=none] (16) at (5.5, 1.25) {};
		\node [style=none] (17) at (5.875, -0.75) {};
		\node [style=hbox] (18) at (1.85, 0.6) {$D$};
		\node [style=none] (19) at (-11, -1.25) {};
		\node [style=none] (20) at (-10, -1.25) {};
		\node [style=none] (21) at (-10.5, -1) {$\cdots$};
		\node [style=rn] (22) at (-10.5, 0) {};
		\node [style=none] (23) at (-11, 1.25) {};
		\node [style=none] (24) at (-10, 1.25) {};
		\node [style=none] (25) at (-10.5, 1) {$\cdots$};
		\node [style=none] (26) at (-11.75, 1.25) {};
		\node [style=none] (27) at (-9, 0) {$=$};
		\node [style=none] (28) at (-7, -1.25) {};
		\node [style=none] (29) at (-6, -1.25) {};
		\node [style=none] (30) at (-6.5, -1) {$\cdots$};
		\node [style=rn] (31) at (-6.5, 0) {};
		\node [style=none] (32) at (-7, 1.25) {};
		\node [style=none] (33) at (-6, 1.25) {};
		\node [style=none] (34) at (-6.5, 1) {$\cdots$};
		\node [style=none] (35) at (-7.75, 1.25) {};
		\node [style=none] (36) at (-7.375, -0.75) {};
		\node [style=hbox] (37) at (-7.75, 0.25) {$D$};
		\node [style=none] (38) at (-4.75, 0) {$=$};
		\node [style=none] (39) at (-2.5, -1.25) {};
		\node [style=none] (40) at (-1.5, -1.25) {};
		\node [style=none] (41) at (-2, -1) {$\cdots$};
		\node [style=rn] (42) at (-2, 0) {};
		\node [style=none] (43) at (-2.5, 1.25) {};
		\node [style=none] (44) at (-1.5, 1.25) {};
		\node [style=none] (45) at (-2, 1) {$\cdots$};
		\node [style=none] (46) at (-3.5, 1.25) {};
		\node [style=none] (47) at (-3.375, -1) {};
		\node [style=hbox] (48) at (-2.875, -0.5) {$D$};
	\end{pgfonlayer}
	\begin{pgfonlayer}{edgelayer}
		\draw [in=45, out=-90] (5.center) to (3);
		\draw [in=135, out=-90] (4.center) to (3);
		\draw [in=-45, out=90] (1.center) to (3);
		\draw [in=-135, out=90] (0.center) to (3);
		\draw [in=45, out=-90] (14.center) to (12);
		\draw [in=135, out=-90] (13.center) to (12);
		\draw [in=-45, out=90] (10.center) to (12);
		\draw [in=-135, out=90] (9.center) to (12);
		\draw [in=-150, out=0, looseness=0.50] (17.center) to (12);
		\draw [in=-90, out=180, looseness=0.50] (17.center) to (16.center);
		\draw [in=-90, out=150] (3) to (7.center);
		\draw [in=45, out=-90] (24.center) to (22);
		\draw [in=135, out=-90] (23.center) to (22);
		\draw [in=-45, out=90] (20.center) to (22);
		\draw [in=-135, out=90] (19.center) to (22);
		\draw [in=150, out=-90] (26.center) to (22);
		\draw [in=45, out=-90] (33.center) to (31);
		\draw [in=135, out=-90] (32.center) to (31);
		\draw [in=-45, out=90] (29.center) to (31);
		\draw [in=-135, out=90] (28.center) to (31);
		\draw [in=-150, out=0, looseness=0.50] (36.center) to (31);
		\draw [in=-90, out=180, looseness=0.50] (36.center) to (35.center);
		\draw [in=45, out=-90] (44.center) to (42);
		\draw [in=135, out=-90] (43.center) to (42);
		\draw [in=-45, out=90] (40.center) to (42);
		\draw [in=-135, out=90] (39.center) to (42);
		\draw [in=-180, out=0, looseness=0.75] (47.center) to (42);
		\draw [in=-90, out=180, looseness=0.50] (47.center) to (46.center);
	\end{pgfonlayer}
\end{tikzpicture}

%% file: figures/zxw/dboxgcopy.tikz
\begin{tikzpicture}
	\begin{pgfonlayer}{nodelayer}
		\node [style=gn] (0) at (-12.5, -0.375) {};
		\node [style=none] (1) at (-12.5, 1.375) {};
		\node [style=none] (2) at (-13, -1.375) {};
		\node [style=none] (3) at (-12, -1.375) {};
		\node [style=hbox] (4) at (-12.5, 0.625) {$D$};
		\node [style=none] (5) at (-11, 0) {$=$};
		\node [style=hbox] (6) at (-9.75, -0.625) {$D$};
		\node [style=gn] (7) at (-9.25, 0.375) {};
		\node [style=none] (8) at (-9.25, 1.375) {};
		\node [style=none] (9) at (-9.75, -1.375) {};
		\node [style=none] (10) at (-8.75, -1.375) {};
		\node [style=hbox] (11) at (-8.75, -0.625) {$D$};
		\node [style=none] (12) at (-1.5, 1.375) {};
		\node [style=none] (13) at (-2, -1.375) {};
		\node [style=none] (14) at (-2.5, 1.375) {};
		\node [style=none] (15) at (-3.75, 0) {$=$};
		\node [style=hbox] (16) at (-1.5, 0.625) {$D$};
		\node [style=hbox] (17) at (-2.5, 0.625) {$D$};
		\node [style=hbox] (18) at (-5.25, -0.625) {$D$};
		\node [style=gn] (19) at (-2, -0.375) {};
		\node [style=none] (20) at (-4.75, 1.375) {};
		\node [style=none] (21) at (-5.25, -1.375) {};
		\node [style=none] (22) at (-5.75, 1.375) {};
		\node [style=gn] (23) at (-5.25, 0.375) {};
		\node [style=none] (24) at (12.5, -1.375) {};
		\node [style=none] (25) at (8.75, 1.375) {};
		\node [style=hbox] (26) at (13, 0.625) {$D$};
		\node [style=none] (27) at (12, 1.375) {};
		\node [style=none] (28) at (5.75, -1.375) {};
		\node [style=none] (29) at (9.75, 1.375) {};
		\node [style=none] (30) at (5.25, 1.375) {};
		\node [style=none] (31) at (13, 1.375) {};
		\node [style=none] (32) at (4.75, -1.375) {};
		\node [style=none] (33) at (10.75, 0) {$=$};
		\node [style=none] (34) at (3.5, 0) {$=$};
		\node [style=hbox] (35) at (12, 0.625) {$D$};
		\node [style=hbox] (36) at (5.75, -0.625) {$D$};
		\node [style=hbox] (37) at (4.75, -0.625) {$D$};
		\node [style=rn] (38) at (5.25, 0.375) {};
		\node [style=none] (39) at (1.5, -1.375) {};
		\node [style=rn] (40) at (12.5, -0.375) {};
		\node [style=hbox] (41) at (9.25, -0.625) {$D$};
		\node [style=hbox] (42) at (2, 0.625) {$D$};
		\node [style=none] (43) at (2.5, -1.375) {};
		\node [style=none] (44) at (2, 1.375) {};
		\node [style=none] (45) at (9.25, -1.375) {};
		\node [style=rn] (46) at (2, -0.375) {};
		\node [style=rn] (47) at (9.25, 0.375) {};
	\end{pgfonlayer}
	\begin{pgfonlayer}{edgelayer}
		\draw (1.center) to (0);
		\draw [in=90, out=-135] (0) to (2.center);
		\draw [in=90, out=-45] (0) to (3.center);
		\draw (8.center) to (7);
		\draw [in=90, out=-135] (7) to (6);
		\draw (6) to (9.center);
		\draw [in=90, out=-45] (7) to (11);
		\draw (11) to (10.center);
		\draw (21.center) to (23);
		\draw [in=-90, out=135] (23) to (22.center);
		\draw [in=-90, out=45] (23) to (20.center);
		\draw (13.center) to (19);
		\draw [in=-90, out=135] (19) to (17);
		\draw (17) to (14.center);
		\draw [in=-90, out=45] (19) to (16);
		\draw (16) to (12.center);
		\draw (44.center) to (46);
		\draw [in=90, out=-135] (46) to (39.center);
		\draw [in=90, out=-45] (46) to (43.center);
		\draw (30.center) to (38);
		\draw [in=90, out=-135] (38) to (37);
		\draw (37) to (32.center);
		\draw [in=90, out=-45] (38) to (36);
		\draw (36) to (28.center);
		\draw (45.center) to (47);
		\draw [in=-90, out=135] (47) to (25.center);
		\draw [in=-90, out=45] (47) to (29.center);
		\draw (24.center) to (40);
		\draw [in=-90, out=135] (40) to (35);
		\draw (35) to (27.center);
		\draw [in=-90, out=45] (40) to (26);
		\draw (26) to (31.center);
	\end{pgfonlayer}
\end{tikzpicture}